\documentclass[11pt,a4paper]{article}

\usepackage[utf8]{inputenc}
\usepackage[T1]{fontenc}
\usepackage{lmodern}
\usepackage[margin=2.55cm]{geometry}
\usepackage{amsmath,amssymb,amsfonts,amsthm}
\usepackage{booktabs}
\usepackage{tabularx}
\usepackage{array}
\usepackage{multirow}
\usepackage{calc}
\usepackage{microtype}
\usepackage{tikz}
\usetikzlibrary{shapes.geometric, arrows.meta, positioning, calc, fit}
\usepackage{hyperref}

\hypersetup{
    colorlinks=true,
    linkcolor=blue!70!black,
    citecolor=blue!70!black,
    urlcolor=blue!70!black,
    pdftitle={GKS FHE: Security Analysis},
    pdfauthor={John Graham, Soundes Marzougui, Milos Pavlovic, Vladimir Shpilrain}
}

\makeatletter
\renewcommand*\l@section[2]{%
  \ifnum \c@tocdepth >\z@
    \addpenalty\@secpenalty
    \addvspace{0.4em \@plus\p@}%
    \setlength\@tempdima{1.8em}%
    \begingroup
      \parindent \z@ \rightskip \@pnumwidth
      \parfillskip -\@pnumwidth
      \leavevmode \bfseries
      \advance\leftskip\@tempdima
      \hskip -\leftskip
      #1\nobreak\normalfont\leaders\hbox{$\m@th\mkern 4.5mu\hbox{.}\mkern 4.5mu$}\hfill
      \nobreak\hb@xt@\@pnumwidth{\hss\normalfont #2}\par
    \endgroup
  \fi}
\makeatother

\theoremstyle{definition}
\newtheorem{definition}{Definition}[section]
\newtheorem{remark}{Remark}[section]

\theoremstyle{plain}
\newtheorem{theorem}{Theorem}[section]
\newtheorem{lemma}{Lemma}[section]
\newtheorem{proposition}{Proposition}[section]

\newcommand{\Z}{\mathbb{Z}}
\newcommand{\F}{\mathbb{F}}
\newcommand{\negl}{\mathrm{negl}}
\newcommand{\Adv}{\mathrm{Adv}}
\newcommand{\Enc}{\mathrm{Enc}}
\newcommand{\Dec}{\mathrm{Dec}}
\newcommand{\KeyGen}{\mathrm{KeyGen}}

\title{\textbf{\LARGE GKS Fully Homomorphic Encryption:\\
Security Analysis}}
\author{
    \textbf{John Graham}\thanks{\texttt{jgraham@zerofhe.ai}, Zero Inc., Connecticut, USA} \and
    \textbf{Soundes Marzougui}\thanks{\texttt{marzougui.soundes@gmail.com}, Qubit Guard, Berlin, Germany} \and
    \textbf{Milo\v{s} Pavlovi\'c}\thanks{\texttt{milos@zerofhe.ai}, Zero Inc., Belgrade, Serbia} \and
    \textbf{Vladimir Shpilrain}\thanks{\texttt{vshpilrain@ccny.cuny.edu}, CUNY, New York, USA}
}
\date{October 2026}

\begin{document}

\maketitle

\begin{abstract}
All traditional encryption schemes share a fundamental operational vulnerability: encrypted data must be decrypted before it can be used. Fully homomorphic encryption (FHE) eliminates this vulnerability by enabling computation directly on encrypted data, but existing lattice-based FHE schemes impose substantial practical constraints that limit deployment.

This paper presents a formal security analysis of the Gribov-Kahrobaei-Shpilrain (GKS)  private-key FHE scheme \cite{GKS18}. GKS operates over the commutative idempotent ring $S_n \cong \bigoplus_{i=1}^{2^n} \Z_p$, a direct sum of $2^n$ copies of $\Z_p$, and requires no bootstrapping or de-linearization, produces bit-exact decryption, supports native encrypted search, and scales to unbounded depth.

GKS is private-key: no external $\Enc$ oracle and no plaintext is supplied to the evaluator; there is no public slot-to-coordinate map and no labelled plaintext--ciphertext pairs. The in-scope games are therefore single-ciphertext IND-EAV and multi-ciphertext COA; KPA/CPA are not claimed. Single-ciphertext IND-EAV holds after conditioning on the publicly observable rank (equivalently, Hamming weight). Unconditional advantage zero is not claimed. Equal-weight ensembles showed no distinguisher above the sampling noise floor (Appendix~A.2). A proof that the statistical distance is negligible remains open, so we do not claim a completed ITS theorem. We analyze multi-ciphertext COA: the adversary holds many unlabeled ciphertexts and may add and multiply. The residual attacks in that setting are collision harvest (rank-order, kernel-span, then ideal closure). Both require at least one full-pack collision. ZCG (Zero Collision Guarantee) is SIMD packing, an epoch-unique nonce, and an Alice-side duplicate check, and is designed to enforce pack-uniqueness at the plaintext and intermediate layers. It is not yet fully implemented (\S\ref{sec:zcg}). Until it is enforced, COA security rests on collision scarcity, ideal rotation, and data hygiene. We do not claim KPA or CPA.

Relevant attack vectors are discussed in Section~\ref{sec:attack_vectors}. Post-ideal-recovery security in the commutative $\Z_p$ instantiation at default parameters is modest. Two patent-protected variants (composite modulus $\Z_N$, US 10,396,976; non-commutative matrix ring, US 9,942,031) provide stronger post-ideal defense. The scheme's private-key architecture, exact arithmetic, and absence of bootstrapping constitute a qualitatively different design point from lattice-based FHE.

Two companion papers address key management \cite{GKSKey} and performance \cite{GKSMetrics}. Through six mechanisms, GKS provides a framework for managing a broad range of use cases. On commodity hardware with no GPU, double encryption, SIMD, and rings 6/7/8, the companion measurements report approximately $2.4~\mu\text{s/slot}$ encryption, $2~\text{ns/slot}$ addition, $31~\text{ns/slot}$ multiplication, $2.3~\mu\text{s/slot}$ decryption, and expansion $16\times$. Chained multiplication stays at $\sim 31~\text{ns/slot}$ through 1,000,000 iterations \cite{GKSMetrics}.
\end{abstract}

\noindent\textbf{Keywords:} fully homomorphic encryption, private-key FHE, retraction ambiguity, idempotent rings, zero-pattern leakage, security games, AI training, machine learning, secure multi-party computation, SMPC, CKKS, BFV, BGV, RLWE, quantum resistance, bijective encoding, bootstrapping, exact arithmetic, lossless encryption, private search, $\text{IND-CPA}^D$, healthcare data, financial services data.

\vspace{1.5em}
\hrule
\tableofcontents
\vspace{1.5em}

\section{Introduction}
\label{sec:intro}
Fully homomorphic encryption (FHE) enables computation on encrypted data without decryption. Since Gentry's construction \cite{Gen09}, the dominant trajectory has pursued lattice-based public-key schemes---BGV \cite{BGV12}, BFV \cite{Bra12, FV12}, CKKS \cite{CKKS17}---all relying on RLWE and requiring bootstrapping.

The GKS scheme \cite{GKS18} departs fundamentally. It is a private-key scheme over the commutative idempotent ring $S_n \cong \bigoplus_{i=1}^{2^n} \Z_p$, achieving: (1) unlimited multiplicative depth without bootstrapping, (2) native encrypted search, (3) a single algorithm for all operations, (4) bit-exact decryption, and (5) a conditional single-ciphertext IND-EAV result (indistinguishability holds after conditioning on the publicly observable rank; Appendix~\ref{app:theorem1}).

The cryptographic construction analyzed in this work is the GKS fully homomorphic encryption (FHE) scheme introduced by Gribov, Kahrobaei, and Shpilrain. We retain the name GKS to maintain continuity with the existing cryptographic literature. ZeroFHE is the company developing and implementing the construction and the additional security and operational mechanisms described in this work. These mechanisms, including collision prevention and the associated ZCG procedure, are intended to address practical attack paths identified in the course of analyzing and implementing GKS. Accordingly, references to GKS in this paper denote the underlying cryptographic construction, while ZeroFHE refers to its implementation and associated system-level enhancements.

Because the symmetric non-lattice line to which GKS belongs has a poor cryptanalytic track record, \S\ref{sec:prior_schemes} gives a complete accounting of the prior schemes in this family, the published attacks that broke them, and the structural reasons GKS does not fall to those attacks.

This paper provides the first detailed security analysis of the GKS scheme in the ePrint literature, with formal security games, a complete retraction counting proof, preliminary cryptanalysis, and comparative analysis.

\subsection{Contributions}
\label{sec:contributions}
\begin{enumerate}
    \item \textbf{Conditional single-ciphertext IND-EAV:} The advantage is negligible after conditioning on the publicly observable rank (equivalently, Hamming weight); unconditional advantage zero is not claimed. Empirical checks on equal-weight plaintexts found no distinguisher above the sampling noise floor (Appendix~\ref{app:theorem1}).
    \item \textbf{A reframing of IND-CPA in deployment terms:} The scheme is not intended for an exposed encryption-oracle model, and the private-key deployment denies such an oracle to every external party. Under an oracle the scheme would not be IND-CPA secure; the zero-pattern leakage channel is closed by bijective encoding and ideal rotation.
    \item \textbf{A KeyGen description-space count (Appendix~\ref{app:theorem2}):} The conservative figure is 10,305 bits ($533,569$ bits nominally, if the unrestricted triangular family is admissible); either is a description-space count under the counting assumptions stated there, not a cryptographic security level, and not residual uncertainty after $I$ is known.
    \item \textbf{A retraction-ambiguity result for non-ring plaintexts:} Existence of multiple retractions is not KPA indistinguishability. Labelled affine recovery (Appendix~\ref{app:affine}) is out of scope of the intended architecture.
    \item \textbf{Preliminary cryptanalysis} of nine distinct attack strategies.
    \item \textbf{Mitigation for zero-pattern leakage} with an algebraic soundness proof for the bijective encoding option.
    \item \textbf{Mitigations for collision harvest attacks} (same-value kernel harvest: rank order and constructed zeros; \S\ref{sec:zcg}).
    \item \textbf{Formal hardness assumptions} (RI, ZKI, RI-D, TBI) with a two-phase security model: defense-in-depth analysis for the hypothetical post-ideal-recovery case (\S\ref{sec:post_ideal}).
    \item \textbf{Performance metric summary} in \S\ref{sec:perf_summary}.
    \item \textbf{COA analysis of the in-scope residual:} Rank-order and same-value kernel harvest. Recovery of same-value differences under one ideal (Appendix~\ref{app:kernel_span}); ZCG (Zero Collision Guarantee, designed, not yet fully implemented; \S\ref{sec:zcg}) will remove original-plaintext and, in double encryption, intermediate layer duplicates. Retraction-ambiguity figures are description counts, not a COA proof.
\end{enumerate}

\subsection{Paper Organization}
\S\ref{sec:intro}: Introduction. \S\ref{sec:gks_construction}: Construction. \S\ref{sec:security_analysis}: Security Analysis. \S\ref{sec:deployment_options}: Deployment Options. \S\ref{sec:post_ideal}: Post Ideal Recovery Hardness. \S\ref{sec:comparative_analysis}: Comparative Analysis. \S\ref{sec:quantum}: Quantum Resistance. \S\ref{sec:empirical}: Empirical Validation. \S\ref{sec:future_work}: Future Work. \S\ref{sec:conclusion}: Conclusion. Appendix~\ref{app:proofs}: Proofs. Appendix~\ref{app:stat_tests}: Statistical Tests.

\subsection{Summary Table}
\begin{table}[h]
\centering
\caption{Security by attack class}
\label{tab:security_by_class}
\begin{tabularx}{\textwidth}{p{3.2cm}p{3.8cm}X}
\toprule
\textbf{Attack class} & \textbf{Applicability} & \textbf{Security status} \\
\midrule
IND-EAV (single-ciphertext) & In scope & IND-EAV holds after conditioning on the publicly observable rank (equivalently, Hamming weight). Unconditional advantage zero is not yet claimed (Theorem~\ref{thm:indeav}, Appendix~\ref{app:theorem1}). Coordinate-distance tests supply supporting evidence for the equal-weight case only. Uniformity tests are ongoing. \\
\addlinespace
COA (Multi-ciphertext) & In scope & Residual issue: requires a same-value pair under one ideal. Mitigated if ZCG is enforced (\S\ref{sec:zcg}); ZCG is designed, not yet fully implemented. \\
\addlinespace
KPA & Out of scope by architecture (no plaintext access; KPA resistance is not claimed) & Out of scope: no labelled pairs and no $\Enc$ oracle. If those are supplied, affine recovery applies (Appendix~\ref{app:affine}). Not claimed. \\
\addlinespace
IND-CPA & Out of scope by architecture (no encryption oracle) & CPA resistance is not claimed. \\
\addlinespace
$\text{IND-CPA}^D$ & Out of scope by architecture (no decryption oracle) & Exact arithmetic; LM21-style noise channel does not apply. \\
\bottomrule
\end{tabularx}
\end{table}

\section{GKS Construction}
\label{sec:gks_construction}

\subsection{Algebraic Defense Structure}
The scheme's algebraic design \cite{GKS18}:
\begin{itemize}
    \item The \textbf{secret random permutation $\pi$} is applied to the orthogonal idempotent basis before the triangular conversion. Before $I$ is recovered, even if an attacker somehow decomposes $S_r$ into its abstract semisimple structure (product of fields), they still do not know which component corresponds to which plaintext ``slot'' or generator. Wedderburn decomposition yields only the \emph{abstract} ring isomorphism type; the concrete labeling of the basis elements is hidden by $\pi$. After $I$ is recovered, quotient canonicalization supplies a coordinate system in which the secret permutation no longer gates plaintext access; matching the resulting canonical components to plaintext slots is absorbed into the subsequent embedding search (see \S\ref{sec:ri} and \S\ref{sec:where_gks_differs}).
    \item The \textbf{secret embedding $\alpha$} (and the associated isomorphism $\varphi: C/I \cong P$) further obscures how the plaintext ring $P$ sits inside the larger ciphertext ring $C$. Knowing $I$ alone does not reveal this coordinate mapping.
    \item Apply a secret permutation $\pi$ on the orthogonal idempotents, then an invertible triangular transformation whose supports are nested: $f_1 = e_{\pi(1)}$, and for $j \ge 2$, $f_j = e_{\pi(j)} + \sum_{i \in F_j} c_i e_{\pi(i)}$, where the index sets satisfy $F_k \subseteq F_j$ whenever $k \le j$ and each coefficient $c_i \in \{0,1\}$. This nested-support condition guarantees the public multiplication rule $f_j f_k = f_k$ for all $k \le j$.
\end{itemize}
After $I$ is recovered in the commutative $\Z_p$ instantiation, neither of these holds: quotient canonicalization becomes the primary attack path, and \S\ref{sec:ri} and \S\ref{sec:where_gks_differs} show that the decryption map can then be read off directly rather than remaining ambiguous. Patent-protected instantiations ($\Z_N$ composite modulus, US 10,396,976; non-commutative matrix ring, US 9,942,031) provide some additional defense-in-depth in that post-$I$ setting (not yet quantified; see \S\ref{sec:future_work}).

\subsection{The Ring $S_n$}
\[
S_n = \langle x_1, \dots, x_n \mid p \cdot 1 = 0,\; x_i^2 = x_i,\; x_i x_j = x_j x_i \rangle \cong \bigoplus_{i=1}^{2^n} \Z_p.
\]
Dimension $2^n$ over $\Z_p$; for $n=7$, dimension 128, $|S_n| = p^{128}$.

\subsection{Orthogonal Idempotent Basis}
For $F \subseteq \{1,\dots,n\}$, define $e_F = \prod_{i \in F} x_i \cdot \prod_{j \notin F}(1-x_j)$. These $2^n$ elements are pairwise orthogonal idempotents forming a basis. Any sum of distinct $e_F$ is an idempotent, giving $2^{2^n}$ total. Multiplication in this basis is componentwise over $\Z_p$.

\subsection{Notation}
$\lambda$: security parameter. $\mathrm{PPT}$: probabilistic polynomial time. $\negl(\lambda)$: negligible function. $A \leftarrow_\$ S$: uniform sampling. $\Adv_{\mathcal{G}}^X(A)$: advantage of $A$ in game $\mathcal{G}$ for scheme $X$.

\subsection{Key Generation}
\begin{itemize}
    \item \textbf{Step 1 (Plaintext ring):} Select a value of $n$ (in this example 7). Plaintext ring: $S_n \cong \Z_p^{128}$.
    \item \textbf{Step 2 (Ideal generation):} Add generators $x_{n+1},\dots,x_r$ ($r=10$), giving $S_r$ with $\dim 2^{10}=1024$. Select ideal $I \subset S_r$ generated by $(x_m - w_m(x_1,\dots,x_{m-1}))$ for $m=n+1,\dots,r$, where each $w_m$ is a random idempotent of $S_{m-1}$. $\dim(I) = 2^r - 2^n = 896$ (at $n=7$).
    \item \textbf{Step 3 (Basis transformation):} Convert to the orthogonal idempotent basis $\{e_i\}$ of $S_r$ (size $2^r$). Apply a secret permutation $\pi$ to this full set of $2^r$ idempotents, then an invertible nested-support triangular transformation: $f_1 = e_{\pi(1)}$, and for $j \ge 2$, $f_j = e_{\pi(j)} + \sum_{i \in F_j} c_i e_{\pi(i)}$, where the index sets satisfy $F_k \subseteq F_j$ whenever $k \le j$ and each $c_i \in \{0,1\}$. The resulting transformed basis of $S_r$ (size $2^r$) is the public ring presentation; ciphertexts are expressed as $2^r$ coefficients in this basis.
    \item \textbf{Step 4 (Embedding):} Select fully homomorphic embedding $\alpha: \Z_p \to S_n$ determined by $\alpha(1)=u$, an idempotent of $S_n$.
\end{itemize}
Private key: $sk = (\text{generators of } I, \text{permutation } \pi, \text{triangular transform } T, \text{embedding } \alpha)$. In the ciphertext-only model, single-ciphertext hiding holds after conditioning on the publicly observable rank (Appendix~\ref{app:theorem1}). A hidden ideal alone does not make two fixed cosets statistically indistinguishable. Key entropy tests are detailed in Appendix~\ref{app:key_entropy}.

\subsection{Encryption and Decryption}
$\Enc(sk,m)$: $E(m) = \alpha(m) + j$, where $j \leftarrow_\$ I$ (uniform random in ideal).

$\Dec(sk,c)$: Convert to standard basis, substitute $x_j \to w_j$ for $j=r,\dots,n+1$. Apply $\alpha^{-1}$. Decryption is exact.

Homomorphic property: $E(a)+E(b) \equiv E(a+b) \pmod I$; $E(a)\cdot E(b) \equiv E(ab) \pmod I$.

\subsection{The Four Components of Retraction Ambiguity}
\label{sec:four_components}
The private key $sk$ consists of four independent components: the ideal generators that determine $I$, the secret permutation $\pi$, the triangular change-of-basis $T$, and the embedding $\alpha$. We state the bound here; the full counting proof is in Appendix~\ref{app:theorem2}.

\paragraph{Component 1: Permutation $\pi$.}
$\pi$ is a uniform random permutation of the full orthogonal idempotent basis of $S_r$ (size $s = 2^r$) applied before the triangular transform.
\begin{itemize}
    \item Number of choices: $(2^r)!$.
    \item Bit contribution: $\log_2((2^r)!)$.
    \item For the reference parameters $r = 10$: $\log_2(1024!) \approx 8{,}769$ bits.
\end{itemize}

\paragraph{Component 2: Nested-support triangular transformations on the full $S_r$ basis ($s = 2^r$).}
The number of such transformations is at least $2^{2^{r-1}}$ (and far larger when all admissible nested families are counted). We use the conservative lower bound arising from free choices on the longest chain; this already contributes more than 500 bits at $r = 10$. The full unitriangular figure $s(s-1)/2 = 523{,}776$ bits counts the unrestricted family and is not established for the nested-support family (Lemma~\ref{lem:triangular}); we report it only as a nominal figure.

\paragraph{Component 3: Ideal generators.}
Each generator $(x_m - w_m)$ selects $w_m$ from the idempotents of $S_{m-1}$. At $n = 7, r = 10$, the three generators contribute $2^{128} + 2^{256} + 2^{512} = 896$ bits.

\paragraph{Component 4: Embedding $\alpha$.}
Determined by $\alpha(1) = \text{an idempotent of } S_n$. Total: $2^{2^n}$ choices, i.e., $2^n$ bits (128 bits at $n = 7$). When bijective encoding is used, the two trivial idempotents are excluded; the bit-security loss is negligible.

\paragraph{Key-generation sampling space.}
Under the parametrization $n=7, r=10$, the total number of distinct private keys that $\KeyGen$ can produce is at least (using the conservative nested-support bound of Lemma~\ref{lem:triangular} for the triangular component)
\[
N_{\KeyGen} \ge 1024! \times 2^{512} \times 2^{896} \times 2^{128}.
\]
Taking base-2 logarithms yields
\[
\log_2 N_{\KeyGen} \ge 8{,}769 + 512 + 896 + 128 = 10{,}305.
\]
This figure measures the size of the space from which keys are drawn. It is a correct lower bound on the description complexity of a private key and on the cost of na\"ive exhaustive search over the entire key space before any ideal recovery. It is not a claim about residual uncertainty after an attacker has recovered $I$ and observed the public ring presentation. Substituting the nominal unrestricted triangular count ($523{,}776$ bits) gives the $533{,}569$-bit figure shown in Table~\ref{tab:retraction_ambiguity}; that figure is not established for the nested-support family.

\paragraph{Residual ambiguity after ideal recovery.}
Once $I$ is known and the concrete public presentation (the published triangular basis and its multiplication table) has been observed, only those keys that are consistent with both pieces of data remain candidates. In particular:
\begin{itemize}
    \item The ideal-generator component is fixed by the recovered $I$ and contributes 0 residual bits.
    \item Keys that would have produced a different public presentation (different effective basis, or different structure constants) are eliminated.
\end{itemize}
The residual components that survive are therefore the permutation $\pi$, the embedding $\alpha$, and whatever freedom remains in the triangular transform $T$ once the observed multiplication table is required to be reproduced. A conservative, publicly verifiable lower bound uses only the two components that are unambiguously free after the public data are fixed:
\[
\log_2(1024!) + 128 \approx 8{,}769 + 128 = 8{,}897 \text{ bits}.
\]
If the implementation's public multiplication rule admits multiple triangular realizations, the residual is larger by the corresponding transform entropy; if the rule forces a unique $T$ once $\pi$ is chosen, the residual is exactly the figure above. In either case the residual search space remains far beyond exhaustive search. The precise contribution of $T$ is an implementation-level fact that should be confirmed by inspection of the published basis and structure constants.

\paragraph{Important caveat.}
Even the residual count above is a combinatorial description-count, not a proven computational lower bound on the cost of recovering a correct retraction. After $I$ is known, quotient canonicalization supplies a coordinate system on the quotient ring; in the commutative $\Z_p$ instantiation this may allow the correct embedding and labelling to be read off with far less work than brute-force search over the residual space (see \S\ref{sec:ri} and \S\ref{sec:where_gks_differs}). Patent-protected non-commutative and composite-modulus instantiations are intended to obstruct that route; their concrete post-$I$ security remains to be quantified. The present theorem therefore supplies only a combinatorial necessary condition, not a sufficient security guarantee in the post-ideal-recovery regime.

\section{Security Analysis}
\label{sec:security_analysis}

\subsection{Hardness Assumptions}
The security of the GKS scheme rests on four computational assumptions specific to the algebraic structure of $S_r$. We state each formally. None reduces to a standard lattice or number-theoretic problem; \S\ref{sec:rlwe_comparison} discusses this explicitly.

\subsubsection{Two-phase security model}
\label{sec:two_phase}
We explicitly distinguish two adversarial phases, because the effective hardness differs between them:
\begin{itemize}
    \item \textbf{Phase 1 (Pre-ideal-recovery):} The adversary does not know $I$. Security rests on ZKI (Definition~\ref{def:zki}) and TBI (Definition~\ref{def:tbi}). In the strict IND-EAV model, we do not claim a completed ITS theorem. The full combinatorial keyspace is relevant as a measure of total secret state.
    \item \textbf{Phase 2 (Post-ideal-recovery):} Of the four key components in \S\ref{sec:four_components}, the ideal generators are fixed once $I$ is known, and quotient canonicalization collapses most of the permutation and triangular ambiguity (\S\ref{sec:key_comp_indep}, \S\ref{sec:post_ideal}). In the commutative $\Z_p$ instantiation the residual post-$I$ barrier is the $O(p)$ embedding search of \S\ref{sec:post_ideal} (about 30 bits at $p \approx 2^{30}$). The GKS patents have mitigations for this attack, but their effectiveness has yet to be fully quantified.
\end{itemize}

\subsubsection{Zero-Knowledge of Ideal (ZKI)}
\label{sec:zki}
\noindent\emph{Intuition.} An adversary observing ciphertexts cannot recover the ideal $I$ from the public presentation of $S_r$ in the transformed basis. Given the quotient canonicalization analysis above, ZKI is the primary security assumption of the GKS scheme.

\begin{definition}[ZKI Assumption]
\label{def:zki}
\textbf{Instance:} The public ring presentation of $S_r$ in the transformed basis, together with a polynomial number of ciphertexts encrypting non-zero plaintexts. \textbf{Challenge:} Output a basis for the secret ideal $I$. \textbf{Assumption:} For every $\mathrm{PPT}$ adversary $\mathcal{A}$,
\[
\Pr[\mathcal{A}\text{ outputs a basis for } I] \le \negl(\lambda).
\]
\end{definition}

In the strict single-ciphertext model (no homomorphic evaluation, no differences), Lemma~\ref{lem:single_ct_ambiguity} shows that $I$ cannot be identified from one ciphertext in the SIMD case; this holds without computational assumptions. Empirical checks on equal-weight plaintexts found no distinguisher above the sampling noise floor (Appendix~\ref{app:theorem1}).

\paragraph{Boundary condition.}
When the adversary can form differences $E(m) - E(m)$ (or otherwise obtain elements of $I$), the information-theoretic guarantee no longer applies. ZKI then becomes a computational assumption: the adversary holds samples from $I$ but must still invert the secret basis transformation in order to recover a usable basis. Ideal rotation (\S\ref{sec:ideal_rotation}) and bijective encoding (\S\ref{sec:bijective_encoding}) limit the number of such samples that can be collected under a fixed ideal.

\subsubsection{Retraction Indistinguishability under Decryption (RI-D)}
\label{sec:rid}
\noindent\emph{Intuition.} Even if the adversary observes exact decryption outputs (the $\text{IND-CPA}^D$ setting of Li--Micciancio \cite{LM21}), the specific attack vector that breaks CKKS---noise-key correlation leakage---does not exist in GKS.

\begin{definition}[RI-D Assumption]
\label{def:rid}
\textbf{Instance:} The public ring presentation of $S_r$, the ideal $I$, and oracle access to $\Dec(sk, \cdot)$ returning exact plaintexts. \textbf{Challenge:} Distinguish encryptions of $m_0$ from encryptions of $m_1$ given the decryption oracle. \textbf{Assumption:} For all $\mathrm{PPT}$ adversaries $\mathcal{A}$ with access to $\Dec(sk, \cdot)$,
\[
\Adv^{\mathrm{RI}\text{-}\mathrm{D}}(\mathcal{A}) \le \negl(\lambda).
\]
\end{definition}

\paragraph{Structural argument against LM21-class attacks.}
The Li--Micciancio attack exploits: (1) approximate decryption leaking a noise term, (2) noise correlated with the secret key, and (3) bootstrapping preserving the correlation. GKS has none of these: decryption is exact, the ideal element $j \in I$ is algebraically killed (not approximately removed), and there is no bootstrapping. The LM21 attack vector is structurally absent.

\paragraph{Scope and limitations.}
The structural argument eliminates the specific LM21 attack class. It does not rule out all $\text{CPA}^D$ attacks. A decryption oracle provides exact evaluation of the projection map $\rho$ on adversary-chosen inputs. Residual attack surfaces include algebraic probing to reconstruct $\rho$ incrementally and kernel/image queries to identify ideal elements. These are substantially mitigated by the private-key model (the adversary cannot craft valid ciphertexts without knowing $I$, which presupposes ZKI compromise). We state RI-D as an assumption because a complete $\text{CPA}^D$ analysis remains open. In the GKS private-key model, the practical relevance of RI-D is limited: an adversary requires knowledge of $I$ to craft valid ciphertexts for submission to the decryption oracle, which presupposes ZKI is already compromised. RI-D is therefore primarily relevant as a comparative benchmark against CKKS rather than as an independent operational threat.

\subsubsection{Triangular Basis Indistinguishability (TBI)}
\label{sec:tbi}
\noindent\emph{Intuition.} The secret permutation $\pi$ and triangular transformation $T$ cannot be determined from the published transformed basis before $I$ is known. TBI protects the basis transformation and is the algebraic mechanism underlying ZKI.

\begin{definition}[TBI Assumption]
\label{def:tbi}
\textbf{Instance:} The published basis $\{f_1, \dots, f_{2^r}\}$ of $S_r$, which is the image of a secret permutation $\pi$ of the full set of $2^r$ orthogonal idempotents followed by a nested-support triangular transform. \textbf{Challenge:} Recover $\pi$ (a permutation of $2^r$ elements) or any of the triangular coefficients, or any information that reduces the effective keyspace. \textbf{Assumption:} For all $\mathrm{PPT}$ adversaries $\mathcal{A}$,
\[
\Pr[\mathcal{A}\text{ recovers }\pi\text{ or any } c_{ij}] \le \negl(\lambda).
\]
\end{definition}

\paragraph{Evidence.}
The multiplication table satisfies $f_j f_k = f_k$ for $k \le j$. Attack 2 (\S\ref{sec:attack2}) shows this is consistent with every nested-support triangular transform (at least $2^{512}$ at $r = 10$; Lemma~\ref{lem:triangular}). The permutation adds $1024! \approx 8{,}769$ bits of further unconstrained degrees of freedom.

Standard semisimple algebra decomposition algorithms \cite{BP23, Ebe91, EG00, FR85} canonicalize any explicit presentation of a finite-dimensional semisimple algebra over a finite field by computing its primitive idempotent basis. For a 128-dimensional algebra over $\Z_p$ ($p \approx 2^{30}$), these algorithms are tractable. The critical question is whether these algorithms can be applied before $I$ is known---i.e., whether TBI can be broken directly on the full ring $S_r$ without first recovering $I$.

We argue this is substantially harder than post-$I$-recovery canonicalization. $S_r$ itself is a 1024-dimensional semisimple algebra whose primitive idempotent decomposition is already known in principle (it is $\Z_p^{1024}$). The problem is not decomposing $S_r$ but determining \emph{which 128 of the 1024 primitive components constitute the plaintext ring $S_n$ vs. the ideal $I$}. This is a subset-identification problem, not a decomposition problem. The basis transformation hides which components correspond to which.

Running Eberly-style decomposition on the full $S_r$ in the transformed basis would yield primitive idempotents in the published coordinates---but the adversary cannot determine which subset corresponds to $I$ without additional information. This distinguishes the pre-$I$ setting (subset identification: hard) from the post-$I$ setting (quotient canonicalization: tractable) discussed in \S\ref{sec:post_ideal}.

\paragraph{Objective assessment.}
We cannot prove that no algorithm bypasses TBI without first recovering $I$. We identify two specific targets for the Stage 2 cryptanalytic challenge and Stage 3 audit: (i) direct canonical decomposition of $S_r$ in the transformed basis to identify the ideal's components (which would simultaneously bypass TBI and ZKI), and (ii) optimized canonicalization of the 128-dimensional quotient at production parameters to confirm the concrete cost of post-$I$-recovery canonicalization.

\paragraph{Role in the security argument.}
TBI supports ZKI as the algebraic mechanism that hides the ideal's location within $S_r$. Its security value is pre-$I$-recovery only. Once $I$ is known, TBI's protection is rendered inert at the quotient level via canonicalization, and security falls through to the embedding $\alpha$.

\subsubsection{Retraction Indistinguishability (RI)}
\label{sec:ri}
\noindent\emph{Intuition.} Given the public presentation of $S_r$ and the ideal $I$, an adversary cannot efficiently determine which retraction $\rho: S_r \to S_n$ was used for encryption.

\begin{definition}[RI Assumption]
\label{def:ri}
\textbf{Instance generation:} Run $\KeyGen(\lambda)$ to obtain $sk = (\pi, T, \{w_m\}, \alpha)$ and the public ring presentation. The distribution over retractions is induced by $\KeyGen$. \textbf{Challenge:} The challenger runs $\KeyGen$ twice with the same $I$ but independent $(\pi, T, \alpha)$, producing $\rho_0$ and $\rho_1$. It selects $b \leftarrow_\$ \{0,1\}$ and provides $I$ and oracle access to encryption under $\rho_b$. (In the private-key model, this oracle is unavailable to the adversary; we state it in the stronger form for compositional arguments.) \textbf{Assumption:} For all $\mathrm{PPT}$ adversaries $\mathcal{A}$,
\[
\Adv^{\mathrm{RI}}(\mathcal{A}) = \left|\Pr[b'=b] - \frac{1}{2}\right| \le \negl(\lambda).
\]
\end{definition}

Quotient canonicalization and effective post-$I$-recovery hardness are discussed in \S\ref{sec:post_ideal}.

\subsubsection{Independence of Key Components}
\label{sec:key_comp_indep}
Theorem~\ref{thm:key_ambiguity} multiplies four component counts. The algebraic independence holds (Appendix~\ref{app:theorem2}, Lemma~\ref{lem:independence}): $\pi$ and $T$ operate on basis representation, ideal generators on the extension structure, and $\alpha$ on the plaintext ring $S_n$. These operate on structurally disjoint aspects of the key. Ring constraints and multiplication table consistency do not create hidden correlations (Attack 2, \S\ref{sec:attack2}).

\paragraph{Important caveat.}
Independence of key descriptions does not imply independence of key effects after quotient canonicalization. Many distinct $(\pi, T)$ pairs yield the same quotient-level behavior once the canonical idempotent basis is computed. The independence argument of Lemma~\ref{lem:independence} is valid for counting key descriptions; it does not contradict the observation that most of this count becomes presentation redundancy---gauge freedom that disappears after quotient normalization.

The dependency structure follows the two-phase model (\S\ref{sec:two_phase}): ZKI and TBI operate pre-$I$-recovery; post-$I$-recovery defense in the commutative $\Z_p$ instantiation is limited to the $O(p)$ embedding search of \S\ref{sec:post_ideal}.

\subsubsection{Limitations and Open Problems}
A commentator observed that, under the default fixed triangular basis, the map from the published $f$-basis to canonical coordinates is a public prefix-sum. Consequently the $2^n$ plaintext-bearing coordinates are unmasked. For non-uniform plaintext distributions an eavesdropper can therefore recover the multiset of slot values (up to the secret permutation) from a single ciphertext. Uniform ensembles show no distinguisher above the sampling noise floor, and the intermediate layer under double encryption is high-entropy. We are currently (i) quantifying the uniformity of production and intermediate-layer data and (ii) completing the random-triangular-basis implementation on the outer layer, which removes the public prefix-sum shortcut. Until those measurements are finished we treat the fixed-$T$ single-layer configuration as unsuitable for low-entropy plaintexts and recommend double encryption together with ideal rotation and (where possible) random-$T$. The algebraic channel is closed by any of: random triangular basis on the observed layer, a non-commutative instantiation, or a genuine restriction to uniform plaintexts.

\begin{enumerate}
    \item \textbf{No standard-problem reduction.} None of RI, ZKI, RI-D, or TBI reduces to a standard lattice or number-theoretic assumption. Non-reducibility reflects algebraic orthogonality, not weakness (\S\ref{sec:rlwe_comparison}), \cite{Nao03}.
    \item \textbf{No worst-case/average-case reduction.} No analogue of Regev's reduction exists for GKS. Establishing one for the Retraction Recovery Problem or the ideal-recovery problem is explicitly open.
    \item \textbf{Post-$I$-recovery security in the commutative $\Z_p$ setting is modest.} Once $I$ is known, the ideal generators are fixed, quotient canonicalization collapses most of the permutation and triangular ambiguity, and the residual barrier is the $O(p)$ embedding search of \S\ref{sec:post_ideal} (about 30 bits at $p \approx 2^{30}$); for ring-structured plaintexts there is no residual barrier. Canonicalization has not yet been run at GKS production parameters. Consistent with \S\ref{sec:where_gks_differs}, non-commutativity alone does not make these algorithms inapplicable: \cite{FR85, EG00} decompose non-commutative finite-dimensional algebras over fields in polynomial time. The load-bearing post-$I$ barrier is the $\Z_N$ instantiation, which ties quotient decomposition to integer factorization, so the residual post-$I$ security in that instantiation is that of the factoring instance (\S\ref{sec:post_ideal}).
    \item \textbf{ZKI is the critical boundary.} Given item (3), the scheme's practical security depends primarily on preventing $I$-recovery. A single ciphertext does not identify $I$ (Lemma~\ref{lem:single_ct_ambiguity}, SIMD case), and empirical checks on equal-weight plaintexts found no distinguisher above the sampling noise floor (Appendix~\ref{app:theorem1}). Across multiple ciphertexts, pairwise rank gives a key-free equality test (\S\ref{sec:attack9}), so $I$-recovery is bounded by the supply of same-value differences under one ideal, not by ZKI alone. Under operational conditions where $E(0)$ elements may accumulate, the mitigations of \S\ref{sec:mitigations} (ideal rotation, bijective encoding) are essential.
    \item \textbf{Quantification of security against attack vectors.} A quantification of security against the major attack vectors is currently being field tested. No tool analogous to the Lattice Estimator yet exists for the post-ideal embedding-search problem; this remains open alongside the audit.
    \item \textbf{Evidence hierarchy.} In order of strength: (a) a proved single-ciphertext ideal-ambiguity result (Lemma~\ref{lem:single_ct_ambiguity}, SIMD case) and a conditional single-ciphertext IND-EAV result where indistinguishability holds after conditioning on the publicly observable rank (Appendix~\ref{app:theorem1}), with empirical support for balance (Appendix~\ref{app:stat_tests}); (b) combinatorial counting of key descriptions with explicit distributional analysis (Theorem~\ref{thm:key_ambiguity}, Appendix~\ref{app:theorem2})---valid as a count, not as a post-canonicalization security estimate; (c) structural inapplicability of LM21-class attacks to GKS (RI-D); (d) preliminary cryptanalysis of nine attack strategies (\S\ref{sec:attack_vectors}); (e) absence of known algorithms that bypass TBI without first recovering $I$. We do not equate absence of known attacks with evidence of hardness.
\end{enumerate}

\subsection{Threat Model: Adversary Classes}
As a private-key architecture, there are three primary security pillars:
\begin{itemize}
    \item Plaintext and $\Enc$ stay with the data owner; the evaluator is not given $m$ or an $\Enc$ oracle.
    \item Encryption does not expose a public coordinate correspondence between plaintext slots and ciphertext coordinates.
    \item Matched $(m, c)$ pairs are not provided to the evaluator.
\end{itemize}
The GKS deployment model involves three parties with distinct capabilities:
\begin{itemize}
    \item \textbf{External adversary (eavesdropper):} Observes ciphertexts only. Security: single-ciphertext IND-EAV, conditional on the balance lemma (Theorem~\ref{thm:indeav}) and rank conditioning (Appendix~\ref{app:theorem1}). Across multiple observed ciphertexts, the pairwise-rank equality test of \S\ref{sec:attack9} is available to an eavesdropper as well.
    \item \textbf{Cloud operator (Carl):} Holds the encrypted database and performs allowed homomorphic operations. Has no encryption oracle and is not given plaintexts. Self-subtraction of a single ciphertext yields the zero ring element, not a fresh ideal sample. The residual in-scope risk is multi-ciphertext COA: if about $N_h \approx \dim I$ (or $\dim K$) linearly independent same-value differences are available under one compute ideal, they span the ideal (or composed kernel). Pairwise rank already gives a key-free equality test (Appendix~\ref{app:rank_diff}), so the binding constraint is the number of repeated plaintext values under one ideal. Mitigations are collision scarcity / ZCG (designed; \S\ref{sec:zcg}) at the inner layer, hygiene and multiplicity bounds at the original layer, ideal rotation, SIMD, and the encodings of \S\ref{sec:mitigations}. IND-EAV is the single-ciphertext model; it does not by itself bound that harvest.
\end{itemize}

\begin{table}[h]
\centering
\caption{Participant capabilities}
\label{tab:participant_capabilities}
\begin{tabularx}{\textwidth}{p{3.5cm}ccccX}
\toprule
\textbf{Adversary} & \textbf{Sees CTs} & \textbf{Sees Interm.} & \textbf{Enc Oracle} & \textbf{Rep. Queries} & \textbf{Controls Enforced} \\
\midrule
EAV (passive) & Yes & No & No & No & N/A \\
Carl (database passive) & Yes & No & No & No & Yes \\
Carl (active search) & Yes & Match only & No & Yes & Yes \\
Carl (compiled API) & Yes & Possible & No & Yes & Advisory only \\
Carl (client sovereign) & Yes & No\textsuperscript{*} & No & Yes & Yes, Alice's key \\
\bottomrule
\end{tabularx}
\vspace{1ex}
\raggedright\footnotesize\textsuperscript{*}KPA pairs obtained through side channel or leakage, not through encryption oracle access. No external party holds an encryption oracle in the GKS private-key architecture.
\end{table}

\subsection{Identified Attack Vectors}
\label{sec:attack_vectors}
We consider nine attack strategies and argue intractability at recommended parameters.

\subsubsection{Attack 1: Gr\"obner basis computation}
The retraction constraints yield a system with $\sim 10^5$ variables and $\sim 10^9$ equations over $\Z_p$. While the GKS system has significant algebraic structure (triangular form, idempotency constraints, Boolean lattice of subsets), at $n=7$, dimension 1024, no existing Gr\"obner basis solver can process systems at this scale. At research parameters ($n=3$, dimension 64), Gr\"obner basis computation is feasible and is recommended as a validation exercise (\S\ref{sec:future_work}), \cite{Fau99, Fau02}. Our current Gr\"obner-basis formulation is not practical at $n=7$; we have not identified a practical Gr\"obner attack at the recommended parameters.

\subsubsection{Attack 2: Triangular transform recovery}
\label{sec:attack2}
The public multiplication table $f_j f_k = f_k$ ($k \le j$) is consistent with every nested-support triangular transform. It does not uniquely determine the secret support sets or the coefficients, so the triangular transform remains information-theoretically hidden among an exponentially large set of possibilities.

\subsubsection{Attack 3: Meet-in-the-middle}
Requires enumerating $2^{63{,}242}$ partial keys in each half---far beyond feasibility.

\subsubsection{Attack 4: Lattice embedding}
The ``noise'' $e \in I$ has arbitrary $\Z_p$-coefficients (not bounded norm). The retraction is algebraic, not geometric. No known embedding reduces the GKS problem to a lattice problem.

\subsubsection{Attack 5: Statistical distinguishing}
For a fixed key, encryptions of a given plaintext are uniform over a coset of $I$, and distinct cosets are disjoint, so this is not by itself an indistinguishability argument (Appendix~\ref{app:single_ct_ensembles}). Single-ciphertext indistinguishability rests on the balance lemma (Theorem~\ref{thm:indeav}) and conditioning on rank (Appendix~\ref{app:theorem1}). For multiple ciphertexts, the pairwise-rank test of \S\ref{sec:attack9} is a successful distinguisher of equal versus unequal plaintexts. However, all plaintext collisions will be removed with ZCG (\S\ref{sec:zcg}), which is designed but not yet fully implemented.

\subsubsection{Attack 6: Linear constraint accumulation from KPA pairs}
An adversary might hope that each known pair $(m_i, c_i)$ provides a linear equation on the retraction $\rho$, and that collecting enough pairs yields a solvable linear system determining $\rho$.
\paragraph{Analysis.} This attack has a fatal prerequisite for non-ring data: to form the equation $c_i - \alpha(m_i) = j_i \in I$, the adversary must compute $\alpha(m_i)$, which requires the secret embedding $\alpha$. A single pair, taken alone, gives no linear equation in $\alpha$ or $I$. Collectively, however, labelled pairs do: decryption is an affine map on ciphertext coordinates, so once the pairs span the ciphertext space, the decryption map can be fitted directly without recovering $(\pi, T, \alpha)$ individually. Appendix~\ref{app:affine} confirms this at 513 and 2,049 pairs for the tested parameters. The case $m_i = 0$ gives an ideal element immediately (\S\ref{sec:attack7}). KPA is therefore not claimed. Under the GKS architecture the evaluator holds no labelled pairs, so this attack is out of scope operationally (Table~\ref{tab:security_by_class}).

\subsubsection{Attack 7: Post-ideal canonicalization / direct-readout attack}
\label{sec:attack7}
\begin{enumerate}
    \item Reduce the ciphertext modulo $I$. The result is an element of the quotient ring $S_r/I \cong S_n \cong (\Z_p)^{2^n}$.
    \item Because $\alpha$ is a homomorphism, the plaintext $m$ sits inside that quotient as the scalar multiple $m \cdot u$ of the fixed idempotent $u = \alpha(1)$.
    \item In the orthogonal (canonical) basis of the quotient, the element $m \cdot u$ is simply the constant $m$ on the support of $u$ and 0 elsewhere.
    \item The attacker therefore reads $m$ directly from those coordinates. No search over embeddings or retractions is required.
\end{enumerate}
When plaintexts do have ring structure (the embedding $\alpha$ is a full ring homomorphism, as required for delegated computation), a different and stronger attack becomes available. This attack is contingent on $I$-recovery, which cannot be done from a single ciphertext (Lemma~\ref{lem:single_ct_ambiguity}) and which, across many ciphertexts, the operational mitigations (\S\ref{sec:ideal_rotation}--\S\ref{sec:double_encryption}) are designed to prevent.

For $p \approx 2^{30}$ (the default for non-ring data), this attack is feasible---but it applies only when plaintexts have ring structure. For ring-structured plaintexts, \cite[\S 5.1]{GKS18} recommends that $p$ has bit length on the order of the required security parameter $t$ (e.g., a 128-bit prime when $t = 128$). For ring data, double encryption is recommended.

\paragraph{The unique exception: $m_i = 0$.}
Every ring homomorphism maps zero to zero: $\alpha(0) = 0$ regardless of which $\alpha$ was chosen. Therefore $E(0) = \alpha(0) + j = j \in I$, and the adversary obtains a raw element of $I$ for free. This is the only plaintext whose encoding is independent of the secret embedding. Zero-pattern leakage (\S\ref{sec:zero_pattern_leakage}) is a consequence of this algebraic fact. With ZCG and double encryption, zero never enters the plaintext or intermediate space, eliminating this input exception; homomorphically derived zeros are addressed separately (\S\ref{sec:bijective_encoding}, Remark).

\subsubsection{Attack 8: Non-Zero KPA}
A recurring structural risk in this scheme is the exposure of elements of $I$, since a confirmed member of $I$ is the unit of currency for the recovery attack of \S\ref{sec:ideal_recovery_differencing}. Elements of $I$ arise operationally as zeros---values whose plaintext component vanishes, leaving a pure ideal element. We classify the sources:
\begin{enumerate}
    \item \textbf{Actual zeros (plaintext-level):} A plaintext value of zero, encrypted directly, yields $E(0) = j \in I$. This channel is closed by bijective encoding (\S\ref{sec:bijective_encoding}): the encoding map is applied before encryption so that no plaintext is represented as a bare zero in the encrypted domain.
    \item \textbf{Self-subtractions:} The operation $c - c$ on any single ciphertext yields the zero element of the ring. The zero vector itself has rank 0 and therefore supplies no linearly independent sample that can help span $I$. Bijective encoding is irrelevant to this cancellation. Nevertheless, the capability to evaluate subtraction is dangerous: it allows differences of distinct computed values and is the most direct route by which an evaluator can manufacture ideal elements. The proper mitigation is therefore architectural---exclusion or strict cryptographic control of subtraction (and equivalent ideal-generating operations) in the allowed function class (compiled API or client-sovereign enforcement channel)---together with rate limiting and ideal rotation. Non-exposure of intermediates remains necessary but is not sufficient.
    \item \textbf{Manufactured zeros (computation-forced):} An adversary may drive two computation paths to a value known a priori to be zero---for example, via identity-forced constants such as products with a zero factor---and difference them. Such a difference is a genuine member of $I$, recognized by arithmetic alone without any oracle. Two properties address this:
    \begin{itemize}
        \item Under the match-only interface and the absence of an exposed comparison predicate in the compute channel (\S\ref{sec:compartmentalization}), the result of such an operation is never observable: the adversary may force a zero but cannot read the ideal element it produces.
        \item Under double encryption (\S\ref{sec:double_encryption}), any such difference an adversary could observe is the outer-encrypted image of the inner ideal element rather than the element itself. The protection is non-observability and outer masking.
    \end{itemize}
\end{enumerate}
In all three cases the security objective is identical: prevent the adversary from obtaining an observable, confirmed element of $I$. The defenses converge on ZCG (\S\ref{sec:zcg}) and double encryption.

\subsubsection{Attack 9: Rank Order Attacks}
\label{sec:attack9}
The rank-order attack is multi-ciphertext COA and thus in scope. For outers $c_a, c_b$, set $d = c_a - c_b$. If they encrypt the same original value, $d$ lies in $I$ (one layer) or in the composed kernel $K$ (double encryption), independently of whether the inner masks match. Public $\mathrm{rank}(M_d)$ is a possible membership signal. It does not satisfy $\mathrm{rank}(M_d) = \dim I$ for every ideal element; empirically, on the tested keys, same-$m$ and different-$m$ differences separated (Appendices~\ref{app:kernel_span} and \ref{app:rank_diff}). A single pair's rank already separates the two cases, so pairwise rank can be a key-free equality test from the outset; no prior harvest is needed to recognize equal pairs.

Once enough same-$m$ differences are in hand, they span $I$ or $K$ (Appendix~\ref{app:kernel_span}). After that, membership classifies further pairs. Double encryption does not move the only relevant collision to the intermediate. ZCG removes plaintext and intermediate collisions.

Differences from different repeated values all lie in the same $I$ (or $K$), so they combine. A corpus of $N$ ciphertexts under one ideal containing $D$ distinct plaintext values supplies $N - D$ linearly independent same-value differences, and Carl can locate them by pairwise rank. ZCG, which is designed but not yet fully implemented, will remove all collisions (\S\ref{sec:zcg}).

\begin{table}[h]
\centering
\caption{Passive and Active Rank Order Attacks}
\label{tab:rank_order_attacks}
\begin{tabularx}{\textwidth}{p{3.5cm}p{4cm}X}
\toprule
\textbf{Aspect} & \textbf{Passive} & \textbf{Active} \\
\midrule
Source of collisions & Already present in the static corpus & Still must occur; Carl can only hope that derived ciphertexts create new ones or that he can test more pairs. With double encryption the intermediate layer is high entropy with multiple vectors making creation of a duplicate highly unlikely. \\
\addlinespace
Number of differences he can form & Limited by the given corpus & Limited only by how many operations the policy/rate limits allow under one ideal. \\
\addlinespace
Ability to adapt & None & Can choose which pairs to difference next based on previous rank results. \\
\addlinespace
Ability to inject new plaintexts & None & No. \\
\bottomrule
\end{tabularx}
\end{table}

ZCG (\S\ref{sec:zcg}) would remove all collisions. Carl has no $\Enc$ oracle, so he cannot manufacture that set. Rank order is out of scope as CPA/KPA when the set would have to be created or labelled.

\paragraph{Mitigations.}
Mitigations fall into two categories:
\begin{itemize}
    \item \textbf{Architectural elimination of duplicates:} ZCG eliminates duplicate packed plaintexts, and in double encryption also checks the intermediate layer (\S\ref{sec:zcg}). Deployment controls provide operational assurance, but are not load-bearing for collision elimination.
    \item \textbf{Operational monitoring of Carl:} The Client Custody Architecture (\S\ref{sec:client_custody}) and Client Sovereign Architecture (\S\ref{sec:client_sovereign}) monitor and limit evaluator access.
\end{itemize}

\subsection{Security Games}

\subsubsection{IND-EAV Security Game}
\begin{definition}[IND-EAV Game]
\label{def:indeav}
The IND-EAV game for GKS proceeds as follows:
\begin{enumerate}
    \item Challenger runs $\KeyGen$ to obtain $sk$.
    \item Challenger selects $m_0, m_1 \leftarrow_\$ S_n$ and $b \leftarrow_\$ \{0,1\}$.
    \item Challenger computes $c^* = \Enc(sk, m_b)$ and sends $c^*$ to adversary $\mathcal{A}$.
    \item $\mathcal{A}$ outputs $b' \in \{0,1\}$.
\end{enumerate}
Advantage: $\Adv^{\mathrm{IND}\text{-}\mathrm{EAV}}(\mathcal{A}) = \left|\Pr[b'=b] - \frac{1}{2}\right|$.
\end{definition}

\begin{theorem}[Single-ciphertext IND-EAV --- conditional]
\label{thm:indeav}
If the balance lemma (Appendix~\ref{app:theorem1}) holds, the IND-EAV advantage of any adversary, including a computationally unbounded one, is negligible, and it is negligible if balance holds with negligible statistical distance. The balance lemma is open; empirical checks to date show no deviation (Appendix~\ref{app:theorem1}).
\end{theorem}

\begin{lemma}[Single-ciphertext ideal ambiguity]
\label{lem:single_ct_ambiguity}
Suppose SIMD packing uses all $2^n$ slots, so that the encoding is a bijection from slot vectors onto $S_n$. Then for every ciphertext $c$ and every candidate key in the support of $\KeyGen$, $c$ lies in the coset $s + I$ of exactly one $s \in S_n$, and so decrypts to exactly one slot vector. A single ciphertext therefore never rules out a candidate ideal, and no procedure can identify $I$ from $c$ by consistency testing. If the slot vector is uniformly distributed, $c$ is uniform on $S_r$ and statistically independent of the key. The non-SIMD encoding $\alpha(m) = m \cdot u$ is not covered and depends on the balance lemma.
\end{lemma}

\begin{remark}
Lemma~\ref{lem:single_ct_ambiguity} does not establish that the ciphertext ensembles for two messages are statistically close. Distinct cosets of a fixed ideal are disjoint, so indistinguishability can only hold for the mixture over $\KeyGen$. That is the balance lemma and indistinguishability holds after conditioning on the publicly observable rank (Appendix~\ref{app:theorem1}).
\end{remark}

\paragraph{Application to GKS.}
At $n = 7, r = 10$, $S_r \cong \Z_p^{1024}$ and $\dim I = 896$. Lemma~\ref{lem:single_ct_ambiguity} covers the single challenge ciphertext of Definition~\ref{def:indeav}. Ring structure supplies one public invariant: the rank of the multiplication-by-$c$ map, which equals $\dim(\ker \rho)$ plus the plaintext's Hamming weight (Appendix~\ref{app:single_ct_ensembles}, Appendix~\ref{app:rank_single}). The weight is therefore public. Bijective encoding shifts the base by one but does not remove the dependence. For two or more ciphertexts, the same invariant applied to a difference becomes an equality test (\S\ref{sec:attack9}, Appendix~\ref{app:rank_diff}), so none of this extends to multi-ciphertext access.

\subsubsection{IND-KPA Security Game}
\begin{definition}[IND-KPA Game for Non-Zero Plaintexts]
\label{def:indkpa}
\begin{enumerate}
    \item Challenger runs $\KeyGen$.
    \item $\mathcal{A}$ receives polynomially many pairs $(m_i, E(m_i))$ with all $m_i \ne 0$.
    \item Challenger selects $m_0, m_1 \ne 0$ and $b \leftarrow_\$ \{0,1\}$; sends $c^* = \Enc(sk, m_b)$.
    \item $\mathcal{A}$ outputs $b'$. Advantage: $\Adv^{\mathrm{KPA}}(\mathcal{A}) = \left|\Pr[b'=b] - \frac{1}{2}\right|$.
\end{enumerate}
\end{definition}

\subsubsection{Theorem 2}
\begin{theorem}[Combinatorial Key Ambiguity]
\label{thm:key_ambiguity}
For $n=7, r=10$, the number of private-key descriptions that $\KeyGen$ can produce is at least
\[
1024! \times 2^{512} \times 2^{896} \times 2^{128} \ge 2^{10{,}305},
\]
using the conservative nested-support bound of Lemma~\ref{lem:triangular} for the triangular component. If the full unitriangular family is admissible, the nominal count is $2^{533{,}569}$.
\end{theorem}
\begin{proof}
See Appendix~\ref{app:theorem2} for the complete derivation, including the independence argument (Lemma~\ref{lem:independence}), the triangular transform count (Lemma~\ref{lem:triangular}), the ideal generator count (Lemma~\ref{lem:ideal_gen}), the embedding count (Lemma~\ref{lem:embedding}), and composition (Theorem~\ref{thm:composition}).
\end{proof}

This bound counts $\KeyGen$ descriptions before ideal recovery. It is not residual ambiguity after $I$ is known: the ideal-generator component then contributes 0 bits, and quotient canonicalization collapses most of the permutation and triangular components (\S\ref{sec:four_components}, \S\ref{sec:key_comp_indep}). The translation of this combinatorial bound into effective post-$I$-recovery security depends on the hardness of the Retraction Recovery Problem (Definition~\ref{def:rrp}, \S\ref{sec:kpa_nonring}).

\begin{remark}[Combinatorial ambiguity vs.\ computational security]
Theorem~\ref{thm:key_ambiguity} establishes a combinatorial lower bound on the number of valid keys. This is a necessary but not sufficient condition for security: it shows that a brute-force enumeration strategy requires searching a space of at least $2^{10{,}305}$ key descriptions. It does not by itself rule out the existence of structured algebraic shortcuts that might narrow this space without full enumeration. The preliminary cryptanalysis in \S\ref{sec:attack_vectors} considers several such shortcut strategies and argues their intractability, but a formal reduction to a recognized hard problem remains open.
\end{remark}

\subsubsection{KPA Non-ring Plaintexts}
\label{sec:kpa_nonring}
Operationally this game is not instantiated. The theorem below is a retraction-ambiguity statement for non-ring plaintexts; it is not a KPA indistinguishability result (with labelled pairs, affine recovery applies; Appendix~\ref{app:affine}).

\begin{theorem}[Retraction Ambiguity --- Non-Ring Plaintexts]
\label{thm:retraction_ambiguity_nonring}
When plaintexts do not have ring structure (the embedding $\alpha$ is one-to-one but not homomorphic---the common enterprise case), for any non-zero ciphertext $c = E(u)$ with $u \ne 0$ and any candidate plaintext $m \ne 0$, there exist valid retractions $\rho$ such that $\rho(c)$ maps to $m$.
\end{theorem}

\begin{definition}[Retraction Recovery Problem]
\label{def:rrp}
Given the public ring presentation of $S_r$ (in the transformed basis), the ideal $I$ (or a basis thereof), and a ciphertext $c \in S_r$, recover the retraction $\rho: S_r \to S_n$ with $\ker(\rho) = I$ that was used for encryption---or equivalently, recover any of the secret key components $(\pi, T, \alpha)$ sufficient to compute $\rho(c)$.
\end{definition}

In the commutative $\Z_p$ instantiation, the Retraction Recovery Problem reduces after $I$-recovery to the $O(p)$ search of \S\ref{sec:post_ideal}; it is stated here for the patent-protected instantiations, where its hardness is open.

\begin{remark}[Bifurcated security profile]
The effective security of GKS depends on whether plaintexts have ring structure:
\begin{itemize}
    \item \textbf{Non-ring plaintexts (enterprise case):} The description count of Theorem~\ref{thm:key_ambiguity} applies (as a count, not a security level); the $p^{-128} \approx 2^{-3840}$ rarity of accidental ideal elements bounds accidental membership only---not recognizability, since \S\ref{sec:attack9} provides a polynomial-time public membership test.
    \item \textbf{Ring-structured plaintexts (delegated computation):} Effective security is $O(p)$, requiring $p \approx 2^{128}$ for 128-bit security. \cite[\S 5.1]{GKS18} explicitly identifies this bifurcation and recommends parameters accordingly.
\end{itemize}
\end{remark}

\subsubsection{IND-CPA Security: Limitation and Context}
The IND-CPA game presumes a capability the GKS deployment model does not grant. In the IND-CPA game, the adversary is granted oracle access to the encryption function. In the GKS private-key architecture, no external party holds that capability: encryption requires the private key, which never leaves the key holder. We state the position plainly rather than claiming structural immunity---were an encryption oracle exposed, repeated queries would enable accumulation of ideal elements and the scheme would not be IND-CPA secure. Consistent with Kerckhoffs's principle, the algorithm is public and only the key is secret; the absence of the oracle is a property of the deployment, not of the algebra.

With bijective encoding (or, under double encryption, at the intermediate layer), zero never enters the plaintext space, closing the input side of the zero-pattern leakage channel; homomorphically derived zeros are bounded by ideal rotation (\S\ref{sec:bijective_encoding}, Remark). Canonicalization is discussed in \S\ref{sec:post_ideal}.

\paragraph{On the $I$-recovery path.}
The identified path to $I$-recovery requires accumulation of elements of $I$, obtainable via encryption oracle queries (as modeled in the IND-CPA game) or homomorphic derivation (Remark~\ref{rem:diff_zeros}).

\begin{remark}[Differences and zeros]
\label{rem:diff_zeros}
The identity $c - c = 0$ on a single ciphertext yields the zero of the ring. It is not a randomized encryption of zero and it contributes nothing to a basis of $I$. The useful identity is the difference of two independent encryptions of the same plaintext:
\[
E_1(m) - E_2(m) = j_1 - j_2 \in I
\]
(or, under double encryption, in the composed kernel). Carl does not need an $\Enc$ oracle to form that difference if those two outers already exist. He also does not need to know $m$. That is the same-value harvest of \S\ref{sec:attack9} / Appendix~\ref{app:kernel_span} (multi-ciphertext COA), not a CPA query of $E(0)$. What he cannot do, without an oracle or a structural same-$m$ set, is manufacture many independent encryptions of one chosen value. Self-subtraction of a lone ciphertext does not replace that.
\end{remark}

\begin{theorem}[IND-CCA2 Impossibility, folklore]
\label{thm:indcca2}
No FHE scheme achieves IND-CCA2.
\end{theorem}
\begin{proof}
Adversary computes $c' = c + \Enc(0) \ne c$, submits to decryption oracle, and recovers $m_b$. Applies to every FHE scheme \cite{Gen09}.
\end{proof}

\subsubsection{Formal KPA-Zero Security Game}
\begin{definition}[IND-KPA-Zero Game]
\label{def:indkpa0}
\begin{enumerate}
    \item Challenger runs $\KeyGen$.
    \item $\mathcal{A}$ receives polynomially many pairs $(m_i, E(m_i))$ where some $m_i$ may equal 0.
    \item Challenger selects $b \leftarrow_\$ \{0,1\}$ and a fresh ciphertext $c^* = \Enc(sk, m_b)$ where $m_0 = 0$ and $m_1 \ne 0$.
    \item $\mathcal{A}$ outputs $b'$. Advantage: $\Adv^{\mathrm{KPA}\text{-}0}(\mathcal{A}) = \left|\Pr[b'=b] - \frac{1}{2}\right|$.
\end{enumerate}
\end{definition}

\paragraph{Observation.}
If $\mathcal{A}$ has accumulated $\ge \dim(I) = 896$ encryptions of zero at $n=7$, $\mathcal{A}$ can reconstruct $I$ via linear algebra and test $c^* \in I$, winning with advantage $1/2$. This is the zero-pattern leakage channel. The mitigations below bound $\mathcal{A}$'s effective advantage.

\subsubsection{Definition}
\label{sec:zero_pattern_leakage}
\begin{definition}[Zero-Pattern Leakage]
An adversary who has recovered $I$ can test $c \in I$, revealing whether $c$ encrypts zero.
\end{definition}
Zero-pattern leakage is one of the residual channels; the principal in-scope channel is the same-value harvest of \S\ref{sec:attack9}. Zero-pattern leakage alone does not reveal non-zero plaintext values. Recovery of $I$ does, however, enable decryption in the commutative $\Z_p$ instantiation via quotient canonicalization (\S\ref{sec:attack7}, \S\ref{sec:post_ideal}), so the relevant bound is on $I$-recovery itself.

\subsubsection{Concrete Bound on Carl's $I$-Recovery Path}
The cloud operator Carl holds $N$ ciphertexts encrypted under ideal $I$ for a given session. Carl has no encryption oracle and no plaintext knowledge. His sole path to accumulating elements of $I$ is collision subtraction: computing $E(a_i) - E(a_j)$ for pairs of ciphertexts that happen to encrypt the same plaintext $a_i = a_j$. When $a_i = a_j$, the result is $j_1 - j_2 \in I$, a nontrivial ideal element. When $a_i \ne a_j$, the result is $\alpha(a_i - a_j) + (j_1 - j_2)$, an element of $S_r$ that is almost surely not in $I$.

Three barriers to $I$-recovery were originally identified; only the first two hold:
\begin{itemize}
    \item \textbf{Barrier 1 (Limited harvest rate):} Carl can perform at most $Q = N(N-1)/2$ pairwise subtractions per session. The expected number yielding ideal elements depends on the plaintext collision rate $p_{\mathrm{coll}}$. For $N$ ciphertexts drawn from $M$ distinct plaintext values, $\mathbb{E}[\text{collisions}] \approx N^2 / (2M)$. SIMD reduces the pure-collision term by the full-vector factor $(2^{-H})^{2^n}$, which makes whole-vector repeats rare (\S\ref{sec:simd}).
    \item \textbf{Barrier 2 (Ideal rotation):} Session-fresh ideals (\S\ref{sec:ideal_rotation}) reset Carl's accumulated samples to zero at each rotation. Samples from session $k$ are elements of $I_k$ and provide zero information about $I_{k+1}$. Carl must harvest $\dim I$ linearly independent ideal elements within a single session (896 at $n = 7, r = 10$). Under double encryption the target is the composed kernel $K$ ($\dim K = 192$ at 6,7,8; Appendix~\ref{app:kernel_peel}).
    \item \textbf{Barrier 3 (Noisy subspace identification --- ZKI):} Even if Carl's subtraction results include ideal elements, he cannot distinguish them from non-ideal results without knowing $I$. The probability that a uniformly random element of $S_r$ lies in $I$ is $p^{-128} \approx 2^{-3840}$, so noise elements are almost surely outside $I$; distinguishing them, however, does not require a separate membership oracle. As shown in \S\ref{sec:attack9} and Appendix~\ref{app:rank_diff}, a passive evaluator can build an $I$-membership test from public ring multiplication alone by computing the dimension of the principal ideal generated by a candidate element. Barrier 3 as originally stated therefore does not hold; the bound of Proposition~\ref{prop:carl_harvest} instead rests on Barriers 1--2 and on the larger harvest target $\dim K$ under double encryption (\S\ref{sec:double_encryption}).
\end{itemize}

\begin{proposition}[Carl's $I$-Recovery Harvest Bound]
\label{prop:carl_harvest}
Let Carl hold $N$ ciphertexts per session under ideal $I$, with ideal rotation period $L$. Let $p_{\mathrm{coll}}$ denote the probability that a uniformly random pair of ciphertexts encrypts the same plaintext.
\begin{enumerate}
    \item[(a)] Expected ideal elements per session: $\mathbb{E}[|\{r_k \in I\}|] = p_{\mathrm{coll}} \cdot N(N-1)/2$.
    \item[(b)] Concentration (heuristic): Let $X$ be the number of pure ideal elements obtainable by collision subtraction in one session. If $\mathbb{E}[X] \le \dim I / 2$, then by a standard Chernoff (or Markov) bound,
    \[
    \Pr[X \ge \dim I] \le 2^{-\Omega(\dim I)} \quad (\text{or } \le 1/2 \text{ under Markov}).
    \]
    The pairwise collision indicators are not independent, so a Chernoff bound is a heuristic, and only the Markov form is rigorous. Rotation parameters should therefore be set from the realized repetition in the data ($N - D$; \S\ref{sec:attack9}).
    \item[(c)] Necessary condition for $I$-recovery (information-theoretic): If the realized number of linearly independent same-value differences under one ideal, $N - D$, is below $\dim I$ (896 at $n = 7$), Carl cannot span $I$ by collision subtraction, even with a membership oracle.
    \item[(d)] Sufficient condition for security (computational): Under ZKI, and with ideal rotation keeping $N - D$ below $\dim I$ in every session with high probability, Carl's probability of accumulating a basis for $I$ by collision subtraction remains negligible.
\end{enumerate}
\end{proposition}

Additionally, Carl cannot match known plaintexts to ciphertexts in the encrypted database: randomized encryption produces different ciphertexts for the same plaintext, and there is no public slot-to-coordinate map.

\subsubsection{Ideal Recovery via Equivalent-Ciphertext Differencing}
\label{sec:ideal_recovery_differencing}
\paragraph{Attack.}
Let $c_1 = E(m)$ and $c_2 = E(m)$ be two encryptions of a common plaintext $m$ with independent randomness. Both lie in the same coset of $I$, so their difference $c_1 - c_2$ is a member of $I$, and it is non-trivial whenever the two randomizations differ. An adversary who can (i) cause the search procedure to evaluate such differences and (ii) learn which differences correspond to true equalities thereby obtains confirmed, labeled elements of $I$. Since $\dim(I) = 896$ at $n=7, r=10$, collecting on the order of $\dim(I)$ independent confirmed elements and taking their span recovers $I$ by linear algebra over $\Z_p$. The reconstruction is polynomial-time and uses only the solution of a linear system; it does not invoke the structure-computation algorithms of \cite{Ebe91, EG00, FR85}. Recovery of $I$ places the scheme in the post-$I$-recovery regime of \S\ref{sec:post_ideal}.

\paragraph{Relation to historical failure modes.}
This is the same linear-algebraic recovery that broke the prior symmetric non-lattice schemes surveyed in \S\ref{sec:prior_schemes}. The distinction for GKS is that the linear attack does not run directly on the secret components $(\pi, T, \alpha)$; it runs on $I$, and only given confirmed equalities.

\paragraph{Defenses.}
The scheme denies the adversary labeled equalities through three layers:
\begin{enumerate}
    \item \textbf{Match-only interface:} The implementation evaluates the aggregate product of per-term differences and returns only match/no-match; a true match yields a product of zero, and no individual difference is exposed.
    \item \textbf{Integrity-protected execution:} Runtime integrity verification (self-checksumming, refusal to execute when attached/patched) raises the cost of extracting intermediates.
    \item \textbf{Double encryption (\S\ref{sec:double_encryption}):} Any difference an adversary could observe is an outer-encrypted image of the inner ideal element, recoverable only with the outer key.
\end{enumerate}

\paragraph{Residual.}
Defenses (1) and (2) withhold equalities from the search interface, but they do not remove the pairwise-rank equality test of \S\ref{sec:attack9}. Against that test the binding control is limiting repeated values under one ideal (\S\ref{sec:attack9}, Proposition~\ref{prop:carl_harvest}). Defense (3) raises the harvest target from $\dim I$ to $\dim K$ but does not remove original-value equalities (\S\ref{sec:double_encryption}).

\subsection{Mitigations}
\label{sec:mitigations}

\subsubsection{Mitigation 1: Ideal Rotation}
\label{sec:ideal_rotation}
For each query session (or at a chosen rotation interval), sample a fresh ideal $I'$ independently of all previous ideals \cite[\S 8.1, Remark 4]{GKS18}. The new ideal is statistically independent of its predecessors; samples collected under an old ideal give no information about $I'$.

\subsubsection{Mitigation 2: Bijective Encoding}
\label{sec:bijective_encoding}
Map plaintext 0 to a nonzero ring element, eliminating $E(0)$ as a distinct class.

\begin{theorem}[Soundness of Bijective Encoding]
\label{thm:bijective_soundness}
Let $e_1$ be any fixed nonzero idempotent of $S_n$ where $e_1 = 1 - u$ and $u = \alpha(1)$, and define the encoding $\phi: \Z_p \to S_n$ by $\phi(m) = \alpha(m) + e_1$, where $\alpha$ is the homomorphic embedding. Then:
\begin{enumerate}
    \item[(i)] $\phi$ is injective;
    \item[(ii)] $0 \notin \mathrm{im}(\phi)$;
    \item[(iii)] $\Enc(\phi(m)) = \phi(m) + j$ for $j \leftarrow_\$ I$ remains a valid encryption;
    \item[(iv)] the additive homomorphism is preserved up to a known shift:
    \[
    \Enc(\phi(a)) + \Enc(\phi(b)) = \Enc(\phi(a+b) + e_1) \text{ \cite{GKS18}}.
    \]
\end{enumerate}
\end{theorem}
\begin{proof}
See Appendix~\ref{app:bijective_proof}.
\end{proof}

\begin{remark}[Scope of bijective encoding]
Bijective encoding and ideal rotation address different attack channels. Bijective encoding prevents zero from entering the plaintext space (input channel): no KPA-zero pairs can arise. Ideal rotation bounds accumulation of homomorphically derived $E(0)$ elements (computation channel). Both mitigations are needed to eliminate the zero-pattern leakage channel entirely. Bijective encoding removes input zeros ($E(0) = j \in I$). It does not stop collision-subtraction of two independent encryptions of the same encoded value:
\[
E_1(\phi(a)) - E_2(\phi(a)) = j_1 - j_2 \in I.
\]
Self-subtraction of one ciphertext is 0 and is useless. Ideal rotation and multiplicity bounds limit how many useful differences can accumulate under one ideal.
\end{remark}

\subsubsection{Mitigation 3: SIMD Security Amplification}
\label{sec:simd}
SIMD packing has two distinct effects on the rank channel:
\begin{enumerate}
    \item Obtaining a pure ideal element (rank exactly $\dim I$) requires simultaneous equality across all $2^n$ slots of the packed vector. The probability of a usable pure sample scales as $(2^{-H})^{2^n}$. At the reference point $n=5, H=15$ bits, this yields a full-vector collision probability of $2^{-480}$, so that a ten-million-record corpus is expected to contain roughly $2^{-433}$ pure ideal samples. This exponential reduction in the pure-membership harvest rate is the main contribution of SIMD against computed-zero aggregation and the rank-order harvest (\S\ref{sec:attack9}).
    \item One might expect intermediate rank values to reveal the Hamming distance between two slot vectors via $\mathrm{rank}(c_1 - c_2) = \dim I + w(\text{support of slot-wise difference})$. Empirically this does not occur in simple tests (Appendix~\ref{app:distance_staircase}). The observable is therefore an equality test on the full packed vector, not a distance oracle.
    \item Pairwise rank of a difference is an equality test on the full packed vector only in the special case that the two vectors differ in at least one slot and the ideal coordinates are fully masked. In general (and in every corrected probe that used orthogonal-support plaintexts) the rank rises with the Hamming weight of the difference, supplying a weight/distance oracle on the packed vector. Under double encryption the measurement is taken on a high-entropy intermediate; the resulting sample does not map back to a useful constraint on the original low-entropy slots. Ideal rotation further limits the number of pairs observed under one ideal. See Appendix~\ref{app:rank_single}, \ref{app:distance_staircase}, and \ref{app:keygen_ensemble}.
\end{enumerate}

\subsubsection{Mitigation 4: Channel Compartmentalization --- Independent Ideals for Search \& Compute}
\label{sec:compartmentalization}
The scheme separates encrypted search and homomorphic computation into distinct cryptographic channels, each with its own independently generated ideal and ciphertext domain. $I_{\mathrm{search}}$ and $I_{\mathrm{compute}}$ are sampled independently: recovery of one yields no information about the other.

\paragraph{Rationale.}
The search channel is the more exposed: its match predicate supplies the labeled equalities on which ideal recovery depends, whereas the compute channel supplies none through its interface (\S\ref{sec:attack8}). Independent ideals prevent search-side leakage from compromising compute ciphertexts \cite{CGKO06, CK10}.

\paragraph{Consequences for implementation.}
\begin{itemize}
    \item Parameters are tuned per channel: search runs at the heavier parameter set, compute runs at the lighter set.
    \item Ciphertexts must not cross channels: enforced via type-level domain separation.
\end{itemize}

\subsubsection{Mitigation 5: Double Encryption}
\label{sec:double_encryption}
The scheme supports a double-encryption mode where the ciphertext produced by the first (inner) encryption is itself encrypted under a second (outer) instance with an independent ideal:
\[
C = \Phi(c) + j_{\mathrm{out}}, \quad \Phi: S_{r_{\mathrm{in}}} \to S_{r_{\mathrm{out}}}.
\]
Double encryption is the recommended default option for GKS.

\paragraph{Effect on security.}
Under double encryption, any difference an adversary observes is the outer-encrypted image of the inner ideal element, not the inner element itself. Same-original-value differences lie in the composed kernel $K$ ($\dim K > \dim I$) and are recognizable by rank (\S\ref{sec:attack9}, Appendix~\ref{app:rank_diff}). Double encryption raises the harvest target from $\dim I$ to $\dim K$ and removes a clean outer-ideal harvest (Appendix~\ref{app:kernel_peel}).

\paragraph{Effect on correctness.}
After inner encryption, data presented to the outer layer carries no semantic structure and is close to uniform. The probability of accidental collisions between distinct plaintexts is minute, lowering search false positives.

\paragraph{Effect on passive KPA.}
Two encryptions of the same known plaintext carry independent inner randomness, so their outer difference is in $K \supset I_{\mathrm{out}}$, not a clean element of $I_{\mathrm{out}}$ (Observation 1).

\begin{table}[h]
\centering
\caption{Benefits of Double Encryption}
\label{tab:double_encryption_benefits}
\begin{tabularx}{\textwidth}{p{3.8cm}p{3cm}p{3.2cm}X}
\toprule
\textbf{Capability} & \textbf{Single Encryption} & \textbf{Double Encryption} & \textbf{Practical Effect of DE} \\
\midrule
Effective full basis for protecting ideal & Yes ($\mathrm{span\_dim} = \dim I$) & No ($\mathrm{span\_dim} > \dim I$) & Main attack path blocked \\
\addlinespace
Ability to protect/canonicalize outer quotient & Yes & No (wrong ideal) & Decryption path potentially blocked \\
\addlinespace
Key-free equality test & Yes & Still sometimes yes & Intermediate equality harvest closed with ZCG; original $m$ harvest requires data hygiene \\
\addlinespace
Ability to peel the two layers & N/A & No (\texttt{big\_kernel}) & Cannot reduce DE to two independent breaks \\
\addlinespace
Sample complexity to reach weak state & Low ($\sim 1.05$--$1.18 \times \dim I$) & High (holds at counts tested) & Attack becomes more expensive and less conclusive \\
\bottomrule
\end{tabularx}
\end{table}

\paragraph{Masking under Double Encryption: Observations.}
\begin{itemize}
    \item \textbf{Observation 1 (generic inner operands do not leak $I_{\mathrm{out}}$):} Let $D = E_{\mathrm{out}}(c_1^{\mathrm{in}}) - E_{\mathrm{out}}(c_2^{\mathrm{in}})$. The rank satisfies:
    \[
    \mathrm{rank}(D) = \dim I_{\mathrm{out}} + w(\alpha_{\mathrm{out}}(c_1^{\mathrm{in}}) - \alpha_{\mathrm{out}}(c_2^{\mathrm{in}})).
    \]
    The outer rank test has three levels: $\dim I_{\mathrm{out}}$ for identical inner ciphertexts, $\dim K$ for distinct inner encryptions of the same original plaintext, and full rank otherwise.
    \item \textbf{Observation 2 (attacker-controlled equal inner operands relocate attack to $I_{\mathrm{out}}$):} When $c_1 = c_2$ exactly, the outer difference reduces to $j_{\mathrm{out},1} - j_{\mathrm{out},2} \in I_{\mathrm{out}}$. Double encryption relocates the attack from $I_{\mathrm{in}}$ to $I_{\mathrm{out}}$, with harvest threshold $\dim(I_{\mathrm{out}})$. Outer-ideal rotation closes this relocated channel.
\end{itemize}

\paragraph{Query Privacy under Double Encryption (Multi-Party Search).}
The searcher encrypts the query under an outer key held by the searcher alone; the data owner performs the search over a query it cannot decrypt, and returns a result the searcher decrypts. The double-encryption guarantee holds against non-colluding parties.

\subsubsection{Mitigation 6: $\Z_N$ Instantiation and Non-Commutative Rings}
These two mitigations address post-ideal recovery canonicalizations (\S\ref{sec:attack7}) and are covered by GKS patents (US 10,396,976 and US 9,942,031; \S\ref{sec:post_ideal}).

\subsubsection{Mitigation 7: ZCG: Collision Elimination}
\label{sec:zcg}
ZCG (Zero Collision Guarantee) combines:
\begin{itemize}
    \item \textbf{Data Hygiene:} Separate short-term ideals for headers/labels.
    \item \textbf{SIMD:} Packing records by rows rather than columns to minimize duplicates.
    \item \textbf{Nonce:} Reserving one or more SIMD slots for an epoch-unique random nonce (drawn from $[1, q)$).
    \item \textbf{Micro-adjustment:} Adjusting duplicate values at the least significant digits.
    \item \textbf{Duplicity Test:} Client-side hash checks before encryption; re-encoding with a new nonce if a collision occurs.
\end{itemize}
With double encryption, ZCG is applied at both the original-plaintext and intermediate layers, removing the collision prerequisite for rank-order and differencing attacks.

\section{Deployment Options}
\label{sec:deployment_options}
The GKS scheme supports four deployment configurations:

\subsection{Compiled API Mode}
Distributed as an integrity-protected compiled library. Rank evaluation (\S\ref{sec:attack9}) is absent from the API surface. Operational controls (subtraction rate limits, key expiry) are advisory in this tier.

\subsection{Client Custody Architecture}
\label{sec:client_custody}
Data owner encrypts and retains ciphertext. Carl submits evaluation requests; computation executes inside the data-owner-controlled environment. Monitoring includes clock controls, counter controls, timer controls, and type controls.

\begin{figure}[h]
\centering
\begin{tikzpicture}[
    box/.style={draw=teal!80!black, rounded corners=3pt, line width=1pt, align=center, fill=white},
    alice_box/.style={draw=teal!80!black, rounded corners=6pt, line width=1.2pt, fill=teal!5},
    carl_box/.style={draw=red!70!black, rounded corners=4pt, line width=1.2pt, fill=red!5, align=center},
    arrow/.style={-{Stealth[length=2.5mm]}, line width=1pt, color=teal!80!black}
]
\node[alice_box, minimum width=15cm, minimum height=4.2cm] (alice) at (0, 2) {};
\node[anchor=north west, font=\bfseries\color{teal!80!black}] at (-7.2, 3.9) {Alice / ZeroFHE / trusted partner};
\node[anchor=north west, font=\footnotesize\itshape\color{black!70}] at (-7.2, 3.5) {Ciphertext and key stay here};

\node[box, minimum width=3.8cm, minimum height=1.2cm] (mpk) at (-4.8, 2.5) {\textbf{Master private key}\\\footnotesize Held locally only};
\node[box, minimum width=3.8cm, minimum height=1.2cm] (store) at (0, 2.5) {\textbf{Ciphertext store}\\\footnotesize Encrypted at rest};
\node[box, minimum width=3.8cm, minimum height=1.2cm] (log) at (4.8, 2.5) {\textbf{Access log}\\\footnotesize Logged per request};

\node[draw=teal!90!black, rounded corners=3pt, fill=teal!80!black, text=white, minimum width=14cm, minimum height=1cm, align=center] (engine) at (0, 0.6) {\textbf{Homomorphic compute engine}\\\footnotesize Operates on ciphertext only};

\draw[arrow] (mpk) -- (mpk |- engine.north);
\draw[arrow] (store) -- (store |- engine.north);
\draw[arrow] (log) -- (log |- engine.north);

\node[carl_box, minimum width=6cm, minimum height=1.5cm] (carl) at (0, -2.2) {\textbf{\color{red!70!black}Carl (compute party)}\\\footnotesize Never holds ciphertext};

\draw[arrow] ([xshift=-1.5cm]carl.north) -- ([xshift=-1.5cm]engine.south) node[midway, left, font=\scriptsize, align=right] {Compute and Search\\Requests};
\draw[arrow] ([xshift=1.5cm]engine.south) -- ([xshift=1.5cm]carl.north) node[midway, right, font=\scriptsize] {Results Only};
\end{tikzpicture}
\caption{Client Custody Architecture}
\label{fig:client_custody}
\end{figure}
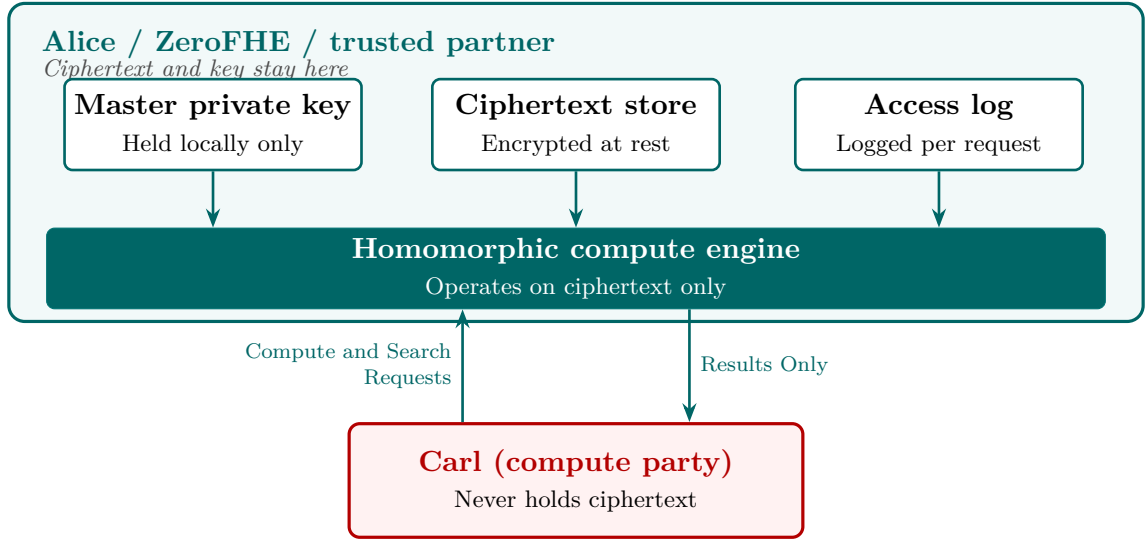

\subsection{Client Sovereign Architecture}
\label{sec:client_sovereign}
All decryption authority remains with the data owner. Three independent ideals serve separate channels: compute ideal, search ideal, and enforcement ideal.

\begin{figure}[h]
\centering
\begin{tikzpicture}[
    scale=0.9, every node/.style={transform shape},
    box/.style={draw=blue!70!black, rounded corners=3pt, line width=0.8pt, align=center, fill=white, minimum width=6cm, minimum height=0.7cm, font=\footnotesize},
    carl_node/.style={draw=orange!80!black, rounded corners=3pt, line width=0.8pt, align=center, fill=orange!5, minimum width=5cm, minimum height=0.7cm, font=\footnotesize},
    zone_alice/.style={draw=blue!70!black, rounded corners=6pt, line width=1pt, fill=blue!2},
    zone_carl/.style={draw=orange!80!black, rounded corners=6pt, line width=1pt, dashed, fill=orange!2},
    arrow/.style={-{Stealth[length=2mm]}, line width=0.8pt, color=black!70}
]
\node[zone_alice, minimum width=7.4cm, minimum height=7.2cm] (alice_zone) at (-4, 0) {};
\node[anchor=north west, font=\bfseries\color{blue!80!black}] at (-7.5, 3.4) {Data owner (Alice)};
\node[anchor=north west, font=\scriptsize\itshape\color{black!60}] at (-7.5, 3.0) {Sovereign zone --- under Alice's key};

\node[box] (a1) at (-4, 2.3) {\textbf{Master private key}\\\tiny Held locally by Alice only};
\node[box] (a2) at (-4, 1.2) {\textbf{Enforcement node}\\\tiny Encrypted counters and clock};
\node[box, draw=orange!90!black, fill=orange!10] (a3) at (-4, 0.1) {\textbf{Capability token}\\\tiny Signed: budget, expiry, scope};
\node[box] (a4) at (-4, -1.0) {\textbf{Encrypted audit log}\\\tiny Logs gate crossings only};
\node[box] (a5) at (-4, -2.1) {\textbf{Ideal rotation}\\\tiny Scheduled + anomaly-triggered};
\node[box] (a6) at (-4, -3.2) {\textbf{Alice decrypts locally}\\\tiny Interprets results, reads log};

\draw[arrow] (a1) -- (a2);
\draw[arrow] (a2) -- (a3);
\draw[arrow] (a3) -- (a4);
\draw[arrow] (a4) -- (a5);

\node[zone_carl, minimum width=6.5cm, minimum height=7.2cm] (carl_zone) at (4, 0) {};
\node[anchor=north west, font=\bfseries\color{orange!80!black}] at (1.0, 3.4) {Carl};
\node[anchor=north west, font=\scriptsize\itshape\color{black!60}] at (1.0, 3.0) {Operator / cloud server};

\node[carl_node] (c1) at (4, 2.3) {\textbf{Bounded operator}\\\tiny No decrypt / forge};
\node[carl_node] (c2) at (4, 1.2) {\textbf{Operation request}\\\tiny Validated by token};
\node[carl_node] (c3) at (4, 0.1) {\textbf{Encrypted data}\\\tiny Ciphertexts only};
\node[carl_node] (c4) at (4, -1.0) {\textbf{Compute layer}\\\tiny Homomorphic ops --- unmetered locally};
\node[carl_node] (c5) at (4, -2.1) {\textbf{Results}\\\tiny Encrypted result};

\draw[arrow] (c2) -- (c3);
\draw[arrow] (c3) -- (c4);
\draw[arrow] (c4) -- (c5);

\draw[arrow] (a3.east) -- (c2.west) node[midway, above, font=\tiny] {token};
\draw[arrow] (c2.west) -- (a3.east) node[midway, below, font=\tiny] {request};
\draw[arrow] (c5.west) -- (a6.east) node[midway, above, font=\tiny] {result};
\end{tikzpicture}
\caption{Client-Sovereign FHE Architecture}
\label{fig:client_sovereign}
\end{figure}
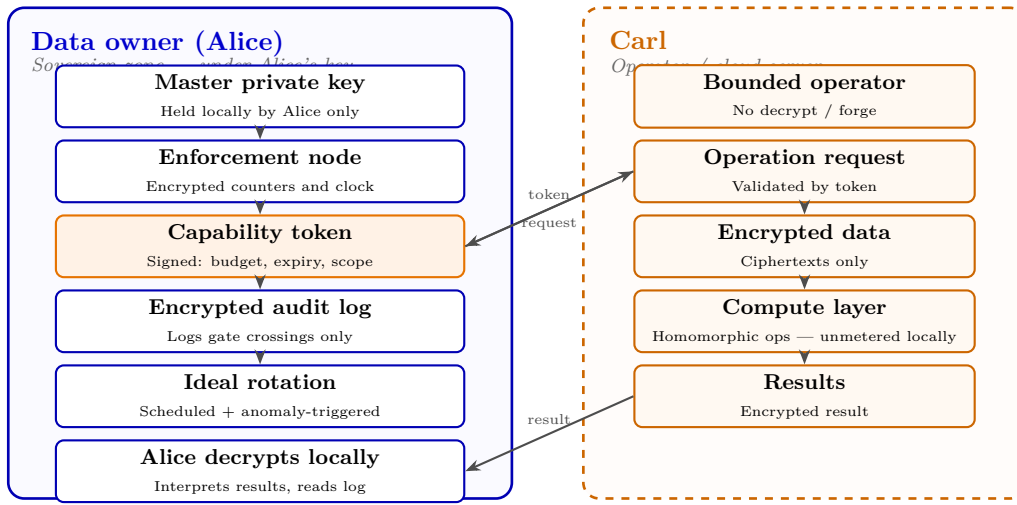

\subsection{Epoch Key Architecture}
\label{sec:epoch_key}
Utilizes ephemeral per-epoch keys whose validity expires at epoch boundaries, gating evaluation capabilities.

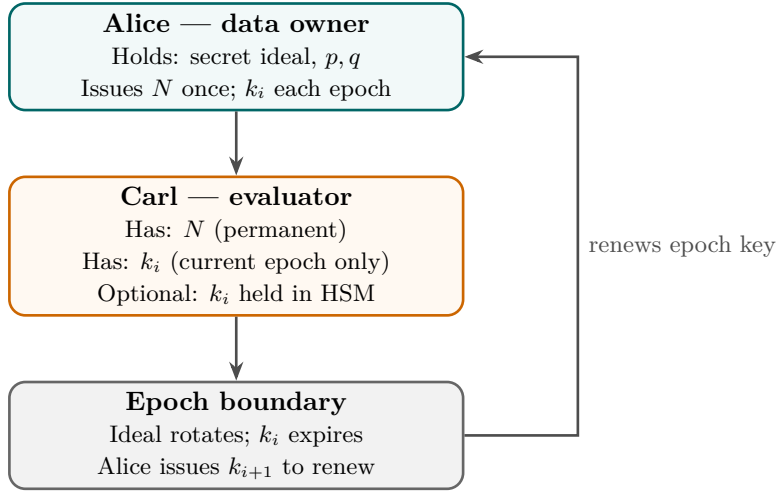
\begin{figure}[h]
\centering
\begin{tikzpicture}[
    box/.style={draw=teal!80!black, rounded corners=6pt, line width=1pt, align=center, fill=teal!5, minimum width=6cm, minimum height=1.3cm, font=\small},
    carl_box/.style={draw=orange!80!black, rounded corners=6pt, line width=1pt, align=center, fill=orange!5, minimum width=6cm, minimum height=1.6cm, font=\small},
    epoch_box/.style={draw=black!60, rounded corners=6pt, line width=1pt, align=center, fill=black!5, minimum width=6cm, minimum height=1.3cm, font=\small},
    arrow/.style={-{Stealth[length=2.5mm]}, line width=1pt, color=black!70}
]
\node[box] (alice) at (0, 3) {\textbf{Alice --- data owner}\\\footnotesize Holds: secret ideal, $p, q$\\\footnotesize Issues $N$ once; $k_i$ each epoch};
\node[carl_box] (carl) at (0, 0.5) {\textbf{Carl --- evaluator}\\\footnotesize Has: $N$ (permanent)\\\footnotesize Has: $k_i$ (current epoch only)\\\footnotesize Optional: $k_i$ held in HSM};
\node[epoch_box] (epoch) at (0, -2) {\textbf{Epoch boundary}\\\footnotesize Ideal rotates; $k_i$ expires\\\footnotesize Alice issues $k_{i+1}$ to renew};

\draw[arrow] (alice) -- (carl);
\draw[arrow] (carl) -- (epoch);
\draw[arrow] (epoch.east) -- ++(1.5,0) |- (alice.east) node[near start, right, font=\footnotesize] {renews epoch key};
\end{tikzpicture}
\caption{Epoch Key Architecture}
\label{fig:epoch_key}
\end{figure}

\section{Post Ideal Recovery Hardness}
\label{sec:post_ideal}
Even in the hypothetical event that $I$ were recovered, an additional exhaustive search of cost $O(p)$ (approximately 30 bits at $p \approx 2^{30}$) would remain. The table below counts the retraction ambiguity in bits for fixed ideals.

\begin{table}[h]
\centering
\caption{Retraction ambiguity (in bits) by parameter set}
\label{tab:retraction_ambiguity}
\begin{tabular}{cccccc}
\toprule
$n$ & $r$ & $\dim S_r$ & Triangular (bits) & Perm (bits) & Total (bits) \\
\midrule
3 & 6 & 64 & 32 & 296 & 392 \\
5 & 8 & 256 & 128 & 1,684 & 2,068 \\
7 & 10 & 1,024 & 512 & 8,769 & 10,305 \\
9 & 12 & 4,096 & 2,048 & 43,250 & 49,394 \\
\bottomrule
\end{tabular}
\end{table}

If $I$ were recovered, an adversary would attempt quotient canonicalization \cite{EG00}. In the commutative $\Z_p$ instantiation, canonicalization runs in polynomial time, leaving an $O(p)$ brute-force search over $\Z_p$ for non-ring data. For ring-structured plaintexts, the plaintext is readable directly from canonical decomposition.

Two patent-protected instantiations provide stronger post-ideal defense:
\begin{itemize}
    \item \textbf{$\Z_N$ Composite Modulus (US 10,396,976):} Ties quotient canonicalization to integer factorization, giving approximately 112 bits of classical hardness at 2048-bit $N$.
    \item \textbf{Non-commutative Matrix Ring (US 9,942,031):} Introduces a non-trivial radical and eliminates direct commutative semisimple decomposition.
\end{itemize}

\section{Comparative Analysis}
\label{sec:comparative_analysis}

\begin{table}[h]
\centering
\caption{Security comparison between FHE schemes}
\label{tab:fhe_comparison}
\begin{tabularx}{\textwidth}{p{3.2cm}p{4cm}p{3.5cm}X}
\toprule
\textbf{Property} & \textbf{GKS} & \textbf{BGV/BFV} & \textbf{CKKS} \\
\midrule
Hardness & Retraction ambiguity (algebraic) & RLWE (lattice) & RLWE (lattice) \\
\addlinespace
IND-EAV & Conditional (balance) & Computational & Computational \\
\addlinespace
IND-CPA & N/A (private-key, no oracle) & Computational & Computational \\
\addlinespace
$\text{IND-CPA}^D$ & Out of scope; LM21 inapplicable & Open & Broken \cite{LM21} \\
\addlinespace
IND-CCA1/2 & Impossible & Impossible & Impossible \\
\addlinespace
KPA (non-ring) & Not claimed & Computational & Computational \\
\addlinespace
KPA (ring) & Not claimed & Computational & Computational \\
\addlinespace
Bootstrapping & None & Required for unbounded depth & Required \\
\addlinespace
Decryption & Exact & Exact (int) & Approximate \\
\addlinespace
Depth & Unlimited & Bounded & Bounded \\
\addlinespace
Key model & Private & Public & Public \\
\bottomrule
\end{tabularx}
\end{table}

See \cite{Cam24} for standardization criteria and concrete parameter bounds for lattice-based schemes.

\subsection{CKKS $\text{IND-CPA}^D$ Vulnerability}
Li--Micciancio \cite{LM21} showed CKKS broken under $\text{IND-CPA}^D$. Exact arithmetic in GKS closes the decryption-error channel structurally.

\subsection{Comparison with RLWE: Different Tradeoffs}
\label{sec:rlwe_comparison}
GKS and RLWE schemes occupy different points in the design space: RLWE benefits from 15+ years of lattice cryptanalysis; GKS offers unlimited depth, exact decryption, and native search based on algebraic assumptions independent of RLWE.

\subsection{Position Relative to Prior Symmetric Non-Lattice FHE}
\label{sec:prior_schemes}
\subsubsection{Prior schemes and their cryptanalysis}
\begin{table}[h]
\centering
\caption{Past private key FHE attempts}
\label{tab:past_schemes}
\begin{tabularx}{\textwidth}{p{3.8cm}cp{3.2cm}p{3cm}X}
\toprule
\textbf{Scheme} & \textbf{Year} & \textbf{Structure} & \textbf{Attack model} & \textbf{Reference} \\
\midrule
Domingo-Ferrer & 2002, 2006 & Polynomial over $\Z_N$ & KPA, poly-time & \cite{CKN06} \\
Kipnis--Hibshoosh MORE & 2012 & Matrix conjugation over $\Z_N$ & Ciphertext-only & \cite{TL15} \\
Kipnis--Hibshoosh PORE & 2012 & Polynomial variant & KPA key recovery & \cite{VV14} \\
Xiao--Bastani--Yen & 2012 & Matrix-based & KPA & \cite{VV14} \\
Gupta--Sharma & 2013 & Linear algebraic & KPA & \cite{VV14} \\
Brenner et al. & 2014 & SHE + refresh & Refresh fails & \cite{EYE19} \\
Li-Scheme & 2016 & $\F_q \hookrightarrow \F_p$ & KPA & \cite{WZZ18} \\
Li-Scheme (improved) & 2016 & $\F_q \hookrightarrow \F_p$ & COA, lattice reduction & \cite{Qu19} \\
Hariss--Noura--Samhat & 2018 & Matrix variant & COA decryption & \cite{VV19} \\
\bottomrule
\end{tabularx}
\end{table}

\subsubsection{Where GKS differs structurally}
\label{sec:where_gks_differs}
\begin{itemize}
    \item \textbf{Coset randomization vs.\ conjugation orbits:} GKS ciphertexts $c = \alpha(m) + j$ sample independent ideal elements $j \in I$, avoiding shared conjugation orbits.
    \item \textbf{Hidden quotient structure vs.\ direct linear key equations:} Ciphertext operations hold modulo $I$; operations are visible only after projecting to $S_r/I$.
    \item \textbf{Quotient canonicalization as post-$I$ boundary:} In the $\Z_N$ variant, canonicalization requires factoring $N$.
\end{itemize}

\subsubsection{What this section does not establish}
This section rules out specific historical linear-algebraic and orbit attacks by structural exclusion; it is not a standalone security proof.

\subsection{GKS Performance Metrics Summary}
\label{sec:perf_summary}
On commodity hardware with no GPU, double encryption, SIMD, and rings 6/7/8, companion measurements report approximately $2.4~\mu\text{s/slot}$ encryption, $2~\text{ns/slot}$ addition, $31~\text{ns/slot}$ multiplication, $2.3~\mu\text{s/slot}$ decryption, and expansion $16\times$ \cite{GKSMetrics}.

\section{Quantum Resistance}
\label{sec:quantum}
In the single-ciphertext model, quantum resistance follows from Theorem~\ref{thm:indeav}. Beyond IND-EAV:
\begin{itemize}
    \item \textbf{Shor \cite{Sho94}:} Inapplicable; GKS does not reduce to abelian hidden subgroup problems.
    \item \textbf{Grover \cite{Gro96}:} Provides quadratic speedup for post-$I$ exhaustive search ($O(\sqrt{p})$).
    \item \textbf{Lattice algorithms \cite{CDPR16}:} Inapplicable (non-lattice).
\end{itemize}

\section{Empirical Validation: Machine Learning on Encrypted Data}
\label{sec:empirical}
\begin{table}[h]
\centering
\caption{Empirical validation of GKS-FHE on medical data}
\label{tab:medical_validation}
\begin{tabularx}{\textwidth}{p{3.5cm}p{3cm}p{3.5cm}X}
\toprule
\textbf{Study} & \textbf{Data} & \textbf{Task} & \textbf{Correctness Finding} \\
\midrule
Gribov et al.\ \cite{GHGK19} & Diabetic (101,767), Parkinson's (5,876) & ML function eval & Zero error discrepancy vs.\ unencrypted; no false positives in $10^6$ multiplications. \\
\addlinespace
Wood et al.\ \cite{WSN19} & Breast Cancer (683 patients) & Naive Bayes classification & Identical accuracy (96.0\%), sensitivity (93.4\%), specificity (97.4\%). \\
\addlinespace
Wood et al.\ \cite{WSNK19} & Breast Cancer & Classification with argmax & Under 1 second; model/data privacy preserved. \\
\bottomrule
\end{tabularx}
\end{table}

\section{Future Work}
\label{sec:future_work}
\subsection{GKS Performance Metrics}
GPU benchmarks are planned as future work.

\subsection{GKS Key Management \& Use Cases}
Companion papers treat operating/search keys, ML-KEM transport, and key delegation.

\subsection{Future Research}
Open problems include: algebraic verification of Theorem~\ref{thm:key_ambiguity} for small parameters, public cryptanalytic challenges, third-party audits, and formal reductions for the Retraction Recovery Problem.

\section{Conclusion}
\label{sec:conclusion}
GKS is a noise-free FHE scheme whose security centers on preventing $I$-recovery (ZKI), mitigated by ZCG (\S\ref{sec:zcg}), providing exact decryption and native search without bootstrapping.

\section*{Acknowledgements}
We thank an anonymous external researcher for independently identifying and sharing the rank order vulnerability with us.

\appendix

\section{Proofs}
\label{app:proofs}

\subsection{Single-ciphertext ensembles and the balance condition}
\label{app:single_ct_ensembles}
Let $I$ be the ideal sampled by $\KeyGen$ and write $C_m = \alpha(m) + J$ with $J$ sampled on $I$. For a fixed $I$, the support of $C_m$ is the coset $\alpha(m) + I$. Distinct cosets of a fixed $I$ are disjoint, so the fact that the adversary does not know $I$ does not make two such cosets statistically indistinguishable. Security must be stated for the mixture over $\KeyGen$. The scheme satisfies the balance condition if the statistical distance between the $\KeyGen$ ciphertext ensembles for every admissible pair $m_0, m_1$ is negligible.

The rank of the public multiplication-by-$c$ map is a similarity invariant of the ring and is therefore computable without knowledge of the key. This rank equals $\dim I$ plus the weight of the component of $c$ outside $I$. Consequently, two cosets whose representatives have different outside-weights are distinguishable by rank alone (Appendix~B.6).

\begin{proof}[Proof of Lemma~1]
For any $\KeyGen$ ideal $I$, $S_r = S_n \oplus I$: the retraction substitutes only for $x_{n+1},\dots,x_r$ and fixes $S_n$, giving $S_n \cap I = 0$, with dimensions adding ($2^n + (2^r - 2^n) = 2^r$). With all $2^n$ slots in use, SIMD encoding is a bijection from slot vectors onto $S_n$. Hence for any $c \in S_r$ and any candidate key, there is exactly one $s \in S_n$ with $c - s \in I$, decrypting to exactly one slot vector. Thus $c$ is consistent with every candidate ideal. If the slot vector is uniform, $c$ is uniform on $S_r$.
\end{proof}

\subsection{Theorem 1 --- IND-EAV Security in the Ciphertext-Only Model}
\label{app:theorem1}
\begin{lemma}[Balance]
The ensembles of ciphertexts encrypting two plaintexts $m_0$ and $m_1$ are identical if and only if $m_0$ and $m_1$ have equal Hamming weight.
\end{lemma}
For plaintexts of unequal weight, the ensembles are distinguishable by public rank ($\Delta = 1$). For plaintexts of equal Hamming weight, sampling over 300 independent $\KeyGen$ draws produced no distinguisher above the sampling noise floor ($\approx 1/\sqrt{K}$). Hence Theorem~1 holds conditionally after conditioning on publicly observable rank. For equal-weight plaintexts the measured distance lay at the sampling noise floor ($\approx 1/\sqrt{K}$); no finite sample can establish that the statistical distance is negligible. This remains an algebraic open question.

\subsection{Theorem 2 --- Retraction Ambiguity Bound}
\label{app:theorem2}
The private key $sk = (\pi, T, \{w_m\}, \alpha)$ consists of: (C1) permutation $\pi$ on the full basis of $S_r$; (C2) invertible triangular transformation $T$; (C3) ideal generators $\{x_m - w_m\}$; and (C4) embedding $\alpha: \Z_p \to S_n$.

\begin{lemma}[Independence of Components]
Components (C1)--(C4) are independently selectable.
\end{lemma}
\begin{proof}
(C1)--(C2) act on basis representations, preserving ring structure. $\pi$ permutes and $T$ mixes triangularly. Ideal generators fix $I = \ker(\rho)$, and embeddings $\alpha$ select idempotents in $S_n$. All selections are algebraically orthogonal.
\end{proof}

\begin{lemma}[Permutation Count]
The number of permutations $\pi$ on the full basis ($s = 2^r$) is $s!$. For $r=10$, $\log_2(1024!) \approx 8{,}769$ bits.
\end{lemma}

\begin{lemma}[Triangular Transformation Count]
Nested-support triangular transformations contribute at least $2^{2^{r-1}}$ choices ($512$ bits at $r=10$). Unrestricted unitriangular binary matrices contribute $s(s-1)/2 = 523{,}776$ nominal bits.
\end{lemma}

\begin{lemma}[Ideal Generator Count]
The number of valid ideal generators for $n=7, r=10$ is $2^{128+256+512} = 2^{896}$.
\end{lemma}

\begin{lemma}[Embedding Count]
The number of fully homomorphic embeddings $\alpha: \Z_p \to S_n$ is $2^{2^n} = 2^{128}$ for $n=7$.
\end{lemma}

\subsection{Theorem 6 (Composition --- Total Retraction Ambiguity)}
For $n = 7, r = 10$, the number of compatible private-key descriptions is at least $2^{10{,}305}$ (conservative) or $2^{533{,}569}$ (nominal unrestricted).
\begin{proof}
$\log_2 N \ge 8{,}769 + 512 + 896 + 128 = 10{,}305$.
\end{proof}

\subsection{Theorem 3 --- Retraction Ambiguity for Non-Ring Plaintexts}
\begin{proof}
When $\alpha$ is not a ring homomorphism, retractions map $c$ independently across candidate plaintexts, ensuring combinatorial ambiguity. With labelled pairs, affine recovery applies (Appendix~B.9).
\end{proof}

\subsection{Theorem 4 --- Soundness of Bijective Encoding}
\paragraph{Addition:} $\phi(m) = \alpha(m) + e_1$. Injectivity holds because $\alpha$ is injective. $0 \notin \mathrm{im}(\phi)$ because coordinates on the support of $e_1$ equal 1. Addition yields $\Enc(\phi(a+b) + e_1)$.
\paragraph{Multiplication:} Since $e_1 \cdot u = 0$ and $e_1^2 = e_1$,
\[
\phi(a)\phi(b) = (\alpha(a) + e_1)(\alpha(b) + e_1) = \alpha(ab) + e_1 = \phi(ab).
\]
Multiplication requires no correction term.

\subsection{Theorem 5 --- IND-CCA2 Impossibility}
Follows by submitting $c' = c + \Enc(0) \ne c$ to the decryption oracle.

\section{Statistical Tests}
\label{app:stat_tests}

In this section we run a series of statistical tests to gain an understanding of the properties of GKS  encrypted data.  The labelled/grouped harness used in these tests removes three of the GKS architecture's standing protections - no Enc oracle and no plaintext to the evaluator; no public slot-to-coordinate map; no labelled $(m,c)$ pairs. Those results describe the embedding when the architectural assumptions are dropped, or when a same-value column is already present. They are not a demonstration that a blind evaluator breaks the scheme. We report it in full for transparency for researchers.  It should not be read as a blind-evaluator break.   All these tests are with synthetic data that has a lower entropy than production data has, and with double encryption already had the entropy significantly increased and collisions greatly reduced or removed altogether. All probes were run with bijective encoding off and random triangular basis off and SIMD on unless otherwise stated. Under double encryption, bijective encoding is not needed at the intermediate layer because inner ciphertexts are not bare zeros; zeros in the original plaintext still produce same-value repeats, see Section 3.3.9.  These tests have not yet been run with the $Z_N$ and non-commutative instantiations of GKS described in the patent and it is expected that these instantiations will show significantly greater resistance even after the security architecture pillars are removed.

Results therefore describe the embedding when these security architecture principles are removed (or when a same value column is provided).  Affine is out of scope as it requires a full pile of matched plaintexts (only available from Alice) but rank order is in scope as it just requires ciphertexts.  Recovery of the ideal (or K) needs about $2^r-2^n$ linearly independent same-value differences under one ideal, which can come from any mix of repeated values. Pairwise rank identifies equal pairs individually (Appendix B5), so the attacker does not need a pre-grouped set. Intermediate-layer duplicates would be eliminated by the Zero Collision Guarantee (designed, not yet fully implemented, see Section 3.5.7). Duplicates at the original plaintext level can be mitigated with data hygiene, ideal rotation, shards, duplicate counting and, alternatively, with deployment architectures that monitor or control Carl.
Labelled harness instantiates a game we do not claim.  Lattice schemes claim that game and put the hardness in the noise.    Under the labelled / grouped harness with three major security pillars (no access to plaintext or oracle, no public slot-to-coordinate map) removed, affine and
same-$m$ rank-order recover as predicted. Other probes do not add a distance oracle, a single-ciphertext support leak (except 0), or a peel of $I_{out}$. Operational security is therefore EAV plus COA bounded by repetition under one ideal.

\subsection*{B1: Key Entropy Tests}

Component entropies by parameter set, for reference alongside the evidence table below.

Two tests actually establish that ideal 0 carries exactly 32 bits at order 5, 128 bits at order 7, not the total. The rest of the report is confirmatory --- it shows no additional weaknesses beyond this known structural point, but doesn't change the headline number. A product of component entropies is an \textbf{upper bound} on attacker effort. If a component is separable, the operative figure is that weakest independently-testable piece. Ideal 0 is independently recoverable at its measured entropy (32 bits at order 5, 128 bits at order 7). Recovery of Ideal 0 does not automatically yield the higher ideal generators, which remain independent, but the low-entropy component itself constitutes the operative security floor for that part of the key. For single encryption we recommend at least order 7 for the plaintext ring and order 10 for the cipher ring.

The table below marks which is which.

\label{app:key_entropy}
\begin{table}[h]
\centering
\caption{Component entropies by parameter}
\label{tab:comp_entropies}
\begin{tabular}{lccccc}
\toprule
\textbf{Order} & \textbf{Ideal 0} & \textbf{Ideal 1} & \textbf{Ideal 2} & \textbf{Total} & \textbf{Weakest Ideal} \\
\midrule
5 & 32 & 64 & 128 & 224 bits & 32 bits \\
7 & 128 & 256 & 512 & 896 bits & 128 bits \\
Double enc.\ (order 8) & 256 & 512 & 1024 & 1792 bits & 256 bits \\
\bottomrule
\end{tabular}
\end{table}

\begin{table}[h]
\centering
\caption{Tests on Ideal Entropy}
\label{tab:ideal_entropy_tests}
\begin{tabularx}{\textwidth}{p{4.2cm}cX}
\toprule
\textbf{Test} & \textbf{Keys used} & \textbf{Finding} \\
\midrule
Reconstruction & 50 & Ideal 0 basis reconstructs byte-identically from 32 sampled bits alone. \\
Birthday / collision & 200,000 & 6 exact collisions observed vs.\ 4.66 predicted at 32 bits. \\
Monobit $\chi^2$ / correlation & 3,000 & Bit uniformity across all coordinates of all three ideals. \\
Permutation uniformity & --- & $\pi$'s fixed-point count matches Poisson(1) expectation. \\
Double-encryption audit & 2,000 (+100) & Outer layer (256-bit weakest ideal) corroborated by dead-coordinate sweeps. \\
\bottomrule
\end{tabularx}
\end{table}

\subsection*{B2: Coordinate Distance Tests}

Pairs of ciphertexts were generated from identical plaintexts and from distinct plaintexts, under both single encryption (n=7, r=10) and double encryption (parameters 5-8-11), modulus 1 000 000 007, without bijective encoding or random triangular basis.

\textbf{Coordinate-distance} Measures mean pairwise Hamming and circular-L2 distance between ciphertexts of the same plaintext versus different plaintexts. In all single- and double-encryption configurations the distances are statistically indistinguishable. The test is a necessary but not sufficient condition for rank-based leakage; its failure to separate is consistent with the absence of a simple metric clustering attack.

\label{app:coord_distance}
\begin{table}[h]
\centering
\caption{Hamming distance tests}
\label{tab:hamming_tests}
\begin{tabularx}{\textwidth}{lccc}
\toprule
\textbf{Test Ciphertexts} & \textbf{Hamming (Differing Coord.)} & \textbf{Circular $L_2$} & \textbf{Conclusion} \\
\midrule
Single Same Plaintext & 251.0 / 256 & 4,572,506 & \multirow{2}{*}{Indistinguishable} \\
Single Different Plaintext & 251.0 / 256 & 4,571,320 & \\
\addlinespace
Double Same Plaintext & 2048 / 2048 & 13,062,532 & \multirow{2}{*}{Indistinguishable} \\
Double Different Plaintext & 2048 / 2048 & 13,063,685 & \\
\bottomrule
\end{tabularx}
\end{table}

\subsection*{B3: Kernel Span Test}

This Appendix B3 does not claim a blind break. It measures harvest size and classification power \textbf{once a same-original- grouping is provided}.

\textbf{Harness result.} Under double encryption, differences of ciphertexts that encrypt the same original plaintext span a single composed kernel of both layers. One full-support same-value difference yields an element of the ideal; public multiplication generates the principal ideal, which equals the ideal with overwhelming probability. Pairwise rank identifies such pairs from ciphertext alone (Appendix B5). The only collision bound that closes the attack is zero collisions under the ideal (ZCG). A non-zero cap is not sufficient.

Incremental linear algebra on those same-differences recovers a subspace of the predicted dimension. Membership in the recovered subspace then classifies held-out pairs with no false positives or negatives: same-original differences lie in the subspace; different differences do not. The test does not recover inner plaintext. It shows that same-original outer differences span both ideals at once, so the composed kernel is an equality oracle on the original plaintext after the harvest.

\label{app:kernel_span}
\begin{table}[h]
\centering
\caption{Rank Order Subspace}
\label{tab:rank_order_subspace}
\begin{tabular}{lcc}
\toprule
\textbf{Rank Order / Ring Dimension} & \textbf{7+1+1} & \textbf{5+3+3} \\
\midrule
Subspace dimension & 384 & 2016 \\
Pairs to harvest it & 584 & 2216 \\
\bottomrule
\end{tabular}
\end{table}

Why this is not the default ciphertext-only game. A launch from ciphertext alone is in scope whenever repeated plaintext values exist under one ideal. The harness supplied a pre-grouped same-value set, but a ciphertext-only evaluator does not need one: the rank of a single difference already separates same-value from different-value pairs (Appendix B5), so equal pairs can be found one at a time, and differences from different repeated values combine (\S3.3.9). The harvest condition is on the total number of linearly independent same-value differences under one ideal ($N - D$ for $N$ ciphertexts with $D$ distinct values), not on any single value's multiplicity.

\textbf{Operational controls}
The following prevent a harvestable same-$m$ block:
\begin{itemize}
\item \textbf{Ideal rotation} on a short horizon, especially for labels and other low-cardinality fields, so any local same-$m$ cluster is bound to a retired ideal.
\item \textbf{Data hygiene}: no headers, no column labels, no structural ``this field is the same key'' markers, no distinguished zeros, no exported grouping keys.
\item Keep low-cardinality fields off the long-lived encrypted surface, or isolate them under separate keys.
\item \textbf{Shards / key separation} so a same-value run cannot be assembled across domains into one linear-algebra instance of size $N_h$.
\item \textbf{Duplicate counting tied to parameters}: total repetition bounds (the duplicate-count rule of \S3.5.7) set so that, at the chosen plaintext and outer ciphertext dimension, the number of independent same-value differences across all values ($N - D$) stays below $N_h$ inside one shard and one ideal epoch.
\item \textbf{Monitoring Carl}: alert on attempts to assemble large same-shape difference collections, on unusual rank/span probes, and on queries that look like membership tests against a candidate subspace.
\end{itemize}

Scope: In scope whenever repeated values exist under one ideal. Contained, not eliminated, when rotation, hygiene, no-label/no-zero layout, shards, repetition bounds and monitoring keep $N - D$ below $N_h$ in every ideal epoch. The table measures harvest size given the equal pairs; per Appendix B5, Carl can identify equal pairs from ciphertext alone.

\subsection*{B4: Kernel Peel Test}

\textbf{Kernel-peel} Tests whether the harvested span is closed under the public multiplication and whether it coincides with the outer ideal. Single encryption yields a clean ideal harvest.

Determines whether the big-kernel ($\dim K > \dim I$ and $K$ is an ideal) $\to$ composed kernel is inseparable, double encryption not peelable to two breaks and shows in test results as ``BigKernel''. Double encryption prevents peeling and raises harvest size. It does not turn rank-order into two independent breaks. Security gain is the larger $N_h$ and the loss of a clean $I_{\mathrm{out}}$, not a second hardness assumption.

The results are shown below and include parameters for both single (marked S in Outer Margin column) and double encryption (Order in the table is what in the paper is described as the ring dimension $n$).

The results show that:

For single encryption for each parameter and configuration (Bijective encoding, SIMD, Triangular Basis) the full ideal was harvested.

For every parameter and configuration with double encryption:
\begin{itemize}
\item The harvested span is larger than the outer ideal.
\item The span is itself an ideal (\texttt{ideal\_membership} = 1).
\item Therefore, the probe classifies it as a single, inseparable ``big kernel''.
\item Carl does not recover the outer ideal $I$. There is no clean-ideal-harvest.
\item The two encryption layers have fused into one larger ideal $K$. The double-encryption structure cannot be peeled into two separate, sequential breaks. This is exactly what the probe means by ``not peelable to two breaks''.
\end{itemize}

More tests will be run with larger sample counts. No dedicated active-rank probe is reported here; the active case is limited by the same intermediate-collision bound and is further constrained by the architectures of \S4.2--4.3. Also, rank order attacks with double encryption and its implications are discussed in \S B3.

$K$ is itself an ideal of the outer ring (\texttt{ideal\_membership} = 1). Whether the composed decryption map is a ring homomorphism with kernel $K$ has not yet been analyzed. If it is, quotient canonicalization of the outer ring modulo $K$ would apply after $K$ is recovered, exactly as for $I$ under single encryption. Until this is settled, recovery of $K$ should be treated as equivalent to ideal recovery, and double encryption as raising the harvest threshold rather than adding a post-recovery barrier.

\label{app:kernel_peel}
\begin{table}[h]
\centering
\caption{Kernel Peel Test for Single and Double Encryption}
\label{tab:kernel_peel}
\resizebox{\textwidth}{!}{%
\begin{tabular}{cccccccccccc}
\toprule
\textbf{Order} & \textbf{Inner} & \textbf{Outer} & \textbf{SIMD} & \textbf{Bij.} & \textbf{Triang.} & \textbf{Verdict} & $\dim I$ & $\dim K$ & \textbf{Ideal Memb.} & $N$ & \textbf{Samples} \\
\midrule
5 & 3 & S & No & No & No & clean-ideal-harvest & 224 & 224 & 1 & 256 & 264 \\
5 & 3 & S & No & Yes & No & clean-ideal-harvest & 224 & 224 & 1 & 256 & 264 \\
5 & 3 & S & Yes & No & No & clean-ideal-harvest & 224 & 224 & 1 & 256 & 264 \\
7 & 3 & S & No & No & No & clean-ideal-harvest & 896 & 896 & 1 & 1024 & 936 \\
7 & 3 & S & No & Yes & No & clean-ideal-harvest & 896 & 896 & 1 & 1024 & 936 \\
7 & 3 & S & Yes & No & No & clean-ideal-harvest & 896 & 896 & 1 & 1024 & 936 \\
6 & 3 & 1 & No & No & No & big-kernel & 512 & 960 & 1 & 1024 & 1000 \\
6 & 3 & 1 & No & No & Yes & big-kernel & 512 & 960 & 1 & 1024 & 1000 \\
6 & 3 & 1 & No & Yes & No & big-kernel & 512 & 960 & 1 & 1024 & 1000 \\
6 & 3 & 1 & No & Yes & Yes & big-kernel & 512 & 960 & 1 & 1024 & 1000 \\
6 & 3 & 1 & Yes & No & No & big-kernel & 512 & 960 & 1 & 1024 & 1000 \\
6 & 3 & 1 & Yes & No & Yes & big-kernel & 512 & 960 & 1 & 1024 & 1000 \\
6 & 1 & 3 & No & No & No & big-kernel & 896 & 960 & 1 & 1024 & 1000 \\
6 & 1 & 3 & No & No & Yes & big-kernel & 896 & 960 & 1 & 1024 & 1000 \\
6 & 1 & 3 & No & Yes & No & big-kernel & 896 & 960 & 1 & 1024 & 1000 \\
6 & 1 & 3 & No & Yes & Yes & big-kernel & 896 & 960 & 1 & 1024 & 1000 \\
6 & 1 & 3 & Yes & No & No & big-kernel & 896 & 960 & 1 & 1024 & 1000 \\
6 & 1 & 3 & Yes & No & Yes & big-kernel & 896 & 960 & 1 & 1024 & 1000 \\
6 & 1 & 1 & No & No & No & big-kernel & 128 & 192 & 1 & 256 & 232 \\
6 & 1 & 1 & No & No & Yes & big-kernel & 128 & 192 & 1 & 256 & 232 \\
6 & 1 & 1 & No & Yes & No & big-kernel & 128 & 192 & 1 & 256 & 232 \\
6 & 1 & 1 & No & Yes & Yes & big-kernel & 128 & 192 & 1 & 256 & 232 \\
6 & 1 & 1 & Yes & No & No & big-kernel & 128 & 192 & 1 & 256 & 232 \\
6 & 1 & 1 & Yes & No & Yes & big-kernel & 128 & 192 & 1 & 256 & 232 \\
7 & 2 & 2 & No & No & No & big-kernel & 1536 & 1920 & 1 & 2048 & 1960 \\
7 & 2 & 2 & No & No & Yes & big-kernel & 1536 & 1920 & 1 & 2048 & 1960 \\
7 & 2 & 2 & No & Yes & No & big-kernel & 1536 & 1920 & 1 & 2048 & 1960 \\
7 & 2 & 2 & No & Yes & Yes & big-kernel & 1536 & 1920 & 1 & 2048 & 1960 \\
7 & 2 & 2 & Yes & No & No & big-kernel & 1536 & 1920 & 1 & 2048 & 1960 \\
7 & 2 & 2 & Yes & No & Yes & big-kernel & 1536 & 1920 & 1 & 2048 & 1960 \\
6 & 1 & 2 & No & No & No & big-kernel & 384 & 448 & 1 & 512 & 488 \\
6 & 1 & 2 & No & No & Yes & big-kernel & 384 & 448 & 1 & 512 & 488 \\
6 & 1 & 2 & No & Yes & No & big-kernel & 384 & 448 & 1 & 512 & 488 \\
6 & 1 & 2 & No & Yes & Yes & big-kernel & 384 & 448 & 1 & 512 & 488 \\
6 & 1 & 2 & Yes & No & No & big-kernel & 384 & 448 & 1 & 512 & 488 \\
6 & 1 & 2 & Yes & No & Yes & big-kernel & 384 & 448 & 1 & 512 & 488 \\
\bottomrule
\end{tabular}%
}
\end{table}

\subsection*{B5: Rank Difference}

\textbf{Rank-difference} Computes the rank of the difference of two ciphertexts and tests whether the rank is systematically lower when the plaintexts (or intermediates) are identical. Single and double encryption both separate cleanly, pair by pair: the rank of one difference equals dim I (single) or dim K (double) for equal plaintexts and is full otherwise. This is a key-free equality test available from the first pair; what requires about $N_h \approx dim ~I$ pairs is spanning the ideal (Appendix B3). Collisions at the intermediate level would be eliminated by ZCG (§3.5.7), but data hygiene is required at the original plaintext level together with the other mitigations discussed in §3.3.9. Another route is with access to an encryption oracle: $m$ can be encrypted $N_h$ times.

\label{app:rank_diff}
\begin{table}[h]
\centering
\caption{Rank Difference Test on Synthetic Data}
\label{tab:rank_diff}
\begin{tabular}{lcccccccc}
\toprule
\textbf{Probe} & \textbf{Order} & \textbf{Inner} & \textbf{Outer} & \textbf{Verdict} & $\dim I$ & \textbf{Slots} & \textbf{Rank\_diff} & \textbf{Rank\_same} \\
\midrule
rank-diff & 5 & 3 & S & separates & 224 & 256 & 256 & 224 \\
rank-diff & 7 & 3 & S & separates & 896 & 1024 & 1024 & 896 \\
rank-diff & 6 & 1 & 1 & separates & 128 & 256 & 256 & 192 \\
rank-diff & 6 & 2 & 1 & separates & 256 & 512 & 512 & 448 \\
rank-diff & 7 & 1 & 1 & separates & 256 & 512 & 512 & 384 \\
\bottomrule
\end{tabular}
\end{table}

\subsection*{B6: Rank Single Ciphertext}

Rank-single-ciphertext examines whether the rank of a single ciphertext varies with the support weight of its plaintext. Across all tested configurations, including double encryption, the rank rises with weight: the all-zero plaintext yields rank $\approx \dim I$, partial-support plaintexts yield intermediate ranks, and full-support plaintexts yield full ambient rank. The probe therefore leaks plaintext support weight, not merely the presence of the zero plaintext. Under double encryption the outer rank correlates with the weight of a high-entropy intermediate ciphertext; a weight sample on that intermediate does not map back to a useful constraint on the original low-entropy field, so larger datasets do not become more dangerous. Per-session ideal rotation further caps the number of samples collected under any single ideal.

\label{app:rank_single}
\begin{table}[h]
\centering
\caption{Rank Single Ciphertext}
\label{tab:rank_single}
\resizebox{\textwidth}{!}{%
\begin{tabular}{lcccccccc}
\toprule
\textbf{Probe} & \textbf{Order} & \textbf{Inner} & \textbf{Outer} & \textbf{Verdict} & $\dim I$ & \textbf{Slots} & \textbf{Rank W0} & \textbf{Rank WFull} \\
\midrule
rank-single & 5 & 3 & S & weight-dep. & 224 & 256 & 224 & 256 \\
rank-single & 7 & 3 & S & weight-dep. & 896 & 1024 & 896 & 1024 \\
rank-single & 6 & 1 & 1 & weight-dep. & 128 & 256 & 192 & 256 \\
rank-single & 6 & 2 & 1 & weight-dep. & 256 & 512 & 448 & 512 \\
rank-single & 7 & 1 & 1 & weight-dep. & 256 & 512 & 384 & 512 \\
\bottomrule
\end{tabular}%
}
\end{table}

\subsection*{B7: Sampler Distance}

\textbf{Sampler-distance} Compares the distribution of same-plaintext differences with the uniform distribution on the ambient space (total variation and $\chi^2$). Single encryption yields near-uniform verdicts; double encryption is occasionally biased (at the lowest total margin value of 2). The probe is only a proxy for the quality of the mask over the ideal and does not by itself indicate a recoverable basis.

\label{app:sampler_distance}
\begin{table}[h]
\centering
\caption{Sampler Distance}
\label{tab:sampler_distance}
\begin{tabular}{lcccccccc}
\toprule
\textbf{Probe} & \textbf{Order} & \textbf{In} & \textbf{Out} & \textbf{Verdict} & \textbf{Buckets} & $\chi^2$ & \textbf{Slots} & $\epsilon_{\mathrm{TV}}$ \\
\midrule
sampler-dist & 5 & 3 & S & near-uniform & 256 & 12400.3 & 256 & 0.02788 \\
sampler-dist & 7 & 3 & S & near-uniform & 256 & 15211.2 & 1024 & 0.01538 \\
sampler-dist & 6 & 1 & 1 & biased & 256 & 219883 & 256 & 0.08191 \\
sampler-dist & 6 & 2 & 1 & near-uniform & 256 & 29880 & 512 & 0.02486 \\
sampler-dist & 7 & 1 & 1 & biased & 256 & 287559 & 512 & 0.06615 \\
\bottomrule
\end{tabular}
\end{table}

\subsection*{B8: Distance-Staircase}

\textit{Distance-staircase} plots $\mathrm{rank}(M_d)$ against the number of differing SIMD slots. In all tested configurations, including double encryption, the rank rises monotonically with the number of differing slots (series of nine values from near $\dim I$ to full ambient dimension). The probe therefore leaks approximate Hamming distance between the compared values, not merely equality. Under double encryption the measurement is taken on outer ciphertexts of high-entropy intermediates; the resulting distance sample does not map to a useful constraint on the original low-entropy plaintext slots. Per-session ideal rotation limits the number of such pairs observed under any single ideal.

\label{app:distance_staircase}
\begin{table}[h]
\centering
\caption{Distance Staircase}
\label{tab:distance_staircase}
\begin{tabular}{lcccccccc}
\toprule
\textbf{Probe} & \textbf{Order} & \textbf{In} & \textbf{Out} & \textbf{Verdict} & $\dim I$ & \textbf{Slots} & \textbf{Min Rank} & \textbf{Max Rank} \\
\midrule
distance & 5 & 3 & S & equality-step & 224 & 256 & 224 & 256 \\
distance & 7 & 3 & S & equality-step & 896 & 1024 & 896 & 1024 \\
distance & 6 & 1 & 1 & equality-step & 128 & 256 & 192 & 256 \\
distance & 6 & 2 & 1 & equality-step & 256 & 512 & 448 & 512 \\
distance & 7 & 1 & 1 & equality-step & 256 & 512 & 384 & 512 \\
\bottomrule
\end{tabular}
\end{table}

\subsection*{B9: Affine Harness Test}

This attack is not an operational path under the private-key architecture (no plaintext and no encryption oracle) and is not a game we claim. Matched $(m, c)$ pairs were supplied by the harness.

Under that labelled harness the decryption map is determined once the pairs span the ciphertext subspace. Let $T$ be the ambient outer dimension in bits: $T = 9$ at 7+1+1 and $T = 11$ at 5+3+3. A textbook affine fit then has threshold $2^T + 1$, i.e. 513 and 2049 respectively. Empirical recovery on order-5 keys succeeded at 60 and 180 labelled pairs, consistent with an information-theoretic floor near clearSize + 1 rather than $n + 1$. Rotation windows must be set from clearSize + 1. Empirically, recovery succeeded in 100 of 100 trials at 513 and 2,049 pairs and in none of 100 trials at up to nine pairs fewer.

Partial information can exist at lower rank. Practical defense: deny Carl matched pairs; short key lifetime below measured rank; Client Sovereign / Custody. Algebraic options still open: $\Z_N$, non-commutative rings, epoch rotation. If Carl has no matched pairs and no encryption oracle the labelled game is not instantiated. If pairs are ever supplied under a stable ideal, the cliff is the whole attack.

\label{app:affine}
\begin{table}[h]
\centering
\caption{Affine Test}
\label{tab:affine_test}
\begin{tabular}{lcc}
\toprule
\textbf{Affine Test / Ring Dimensions} & \textbf{7+1+1} & \textbf{5+3+3} \\
\midrule
Pairs needed & 513 & 2049 \\
Up to nine pairs short & 0 / 100 & 0 / 100 \\
At threshold count & 100 / 100 & 100 / 100 \\
Mask annihilation & 25 / 25 & 25 / 25 \\
Forgery accepted & 25 / 25 & 25 / 25 \\
\bottomrule
\end{tabular}
\end{table}

\subsection*{B10: KeyGen-Ensemble Measurements}

This appendix records the KeyGen-ensemble measurements underlying the Balance statement in A.2.

Balance condition: what was measured:

\textbf{Parameters:} Order 5, extra-order 3, q = 1,000,003 (clearSize 32, cipherSize, n = 256), fixed triangular basis. Experiments 1--3 single layer; Experiment 4 repeats the decisive case under double encryption (5+3+3, n = 2048).

\textbf{Tools:} \texttt{Collaboration/key\_security\_analysis/\{balance\_test, balance\_multi, balance\_control, balance\_de\}.cpp}

\textbf{Experimental design:} Appendix~A.2 asks whether, for two fixed plaintexts $m_0$ and $m_1$, the ensembles
\[
\{\Enc(sk, m_0)\}_{sk \leftarrow \KeyGen} \quad \text{and} \quad \{\Enc(sk, m_1)\}_{sk \leftarrow \KeyGen}
\]
are statistically close. The randomness sampled is $\KeyGen$ plus encryption randomness---not repeated encryptions under one key. Each iteration accordingly constructs a fresh \texttt{Context}, a fresh \texttt{KeyGenerator}, and a fresh \texttt{Key}, then encrypts both $m_0$ and $m_1$ under that key. One ciphertext per message per key enters the ensemble. Statistical distance is computed as
\[
\frac{1}{2} \sum_{x} |p_A(x) - p_B(x)|
\]
over the empirical histograms.

Plaintexts are built by SIMD slot (orthogonal basis), so ``weight" means the number of nonzero orthogonal slots.

\label{app:keygen_ensemble}
\begin{table}[h]
\centering
\caption{Statistical Distance of KeyGen Ensembles}
\label{tab:keygen_stat_dist}
\begin{tabular}{lcccc}
\toprule
$m_0$ & $m_1$ & $\mathrm{rank}(m_0)$ & $\mathrm{rank}(m_1)$ & $\Delta$ \\
\midrule
weight 1 & weight 2 & $225 \times 40$ & $226 \times 40$ & 1.00 \\
all-zero & full support & $224 \times 40$ & $256 \times 40$ & 1.00 \\
weight 5 & weight 5 & $229 \times 40$ & $229 \times 40$ & 0.00 \\
weight 16 & weight 16 & $240 \times 40$ & $240 \times 40$ & 0.00 \\
\bottomrule
\end{tabular}
\end{table}

\begin{table}[h]
\centering
\caption{Same Weight, Four Statistics ($K=300$)}
\label{tab:same_weight_stats}
\begin{tabular}{lc}
\toprule
\textbf{Statistic} & \textbf{Result} \\
\midrule
$\mathrm{Rank}(M_c)$ & $\Delta = 0.000$ \\
Zero-coefficient count & $\Delta = 0.000$ \\
$\mathrm{Coefficient}[0] \pmod{97}$ & $\Delta = 0.320$ \\
Mean coefficient over 256 positions & 495,743 vs 494,727 ($z=0.66$) \\
\bottomrule
\end{tabular}
\end{table}

\begin{table}[h]
\centering
\caption{Control --- Noise Levels}
\label{tab:control_noise}
\begin{tabular}{lc}
\toprule
\textbf{Comparison} & $\Delta$ \\
\midrule
Same message, two disjoint key sets & 0.303 (noise floor) \\
$m_0$ vs $m_1$, equal weight & 0.320 \\
\bottomrule
\end{tabular}
\end{table}

\begin{table}[h]
\centering
\caption{Double Encryption ($5+3+3$, $n=2048$, $K=20$)}
\label{tab:de_eval}
\begin{tabular}{lcccc}
\toprule
$m_0$ & $m_1$ & $\mathrm{rank}(m_0)$ & $\mathrm{rank}(m_1)$ & $\Delta$ \\
\midrule
weight 1 & weight 2 & $2017 \times 20$ & $2018 \times 20$ & 1.00 \\
all-zero & full support & $2016 \times 20$ & $2016 \times 20$ & 1.00 \\
weight 5 & weight 5 & $2021 \times 20$ & $2021 \times 20$ & 0.00 \\
\bottomrule
\end{tabular}
\end{table}

\end{document}